\documentclass{conm-p-l} %Document class changes to AMS Standard 6/24/12

\usepackage{graphicx}
\usepackage{mathrsfs, amsthm, amsmath, amssymb}
\usepackage{appendix}
\usepackage{bm, hyperref}
\usepackage{cite}
\usepackage{url}
\usepackage{epsfig} %NEW 8/7/11
\usepackage{euscript} %NEW 8/7/11
\usepackage{amsfonts} %NEW 8/7/11

\newtheorem{theorem}{Theorem}[section] %Modified 3/12/12 to fit MASSM paper
\newtheorem{lemma}[theorem]{Lemma}
\newtheorem{corollary}[theorem]{Corollary}
\newtheorem{conjecture}[theorem]{Conjecture}
\newtheorem{proposition}[theorem]{Proposition}

\theoremstyle{definition}
\newtheorem{definition}[theorem]{Definition}

\newtheorem{example}[theorem]{Example}

\newtheorem{remark}[theorem]{Remark}

\newtheorem{notation}[theorem]{Notation}

\newcommand{\hs}{\hspace{2mm}}

\newcommand{\C}{\mathbb{C}}
\newcommand{\N}{\mathbb{N}}

\newcommand{\R}{\mathbb{R}}
\newcommand{\Z}{\mathbb{Z}}

\newcommand{\bfa}{\mathbf{a}} %NEW 3/31/12
\newcommand{\bfc}{\mathbf{c}} %NEW 4/4/12
\newcommand{\bfg}{\mathbf{g}}

\newcommand{\bfk}{\mathbf{k}}
\newcommand{\bfm}{\mathbf{m}} %NEw 3/12/12
\newcommand{\bfp}{\mathbf{p}}
\newcommand{\bfq}{\mathbf{q}}
\newcommand{\bfr}{\mathbf{r}}
\newcommand{\bfs}{\mathbf{s}} %NEw 3/31/12
\newcommand{\bft}{\mathbf{t}} %NEW 3/28/12
\newcommand{\bfv}{\mathbf{v}} %NEW 7/9/11
\newcommand{\bfPhi}{\mathbf{\Phi}}
\newcommand{\calD}{\mathcal{D}} %NEW 8/7/11
\newcommand{\calL}{\mathcal{L}}
\newcommand{\calR}{\mathcal{R}} %NEW 7/16/11
\newcommand{\calT}{\mathcal{T}} %NEW 4/2/12

\newcommand{\scrJ}{\mathscr{J}} %NEW 3/12/12
\newcommand{\Rec}{\calR} 

\newcommand{\recpoly}[2]{ #1^#2 = a_1#1^{#2-1} + \cdots + a_#2}

\newcommand{\recmatrix}[9]{  %removed 4th row and column
\begin{pmatrix}
 #9  & \cdots  & #5 \\
 #8  & \cdots  & #4 \\
 \vdots & \ddots & \vdots \\
 #5  & \cdots  & #1 
\end{pmatrix} }

\newcommand{\cfmatrix}[4]{ %%%removed 3rd column and 2nd row.
\begin{pmatrix}  
 #1 & 1   & \cdots & 0 \\
\vdots & \vdots  & \ddots & \vdots \\
 #3 & 0   & \cdots & 1 \\
 #4 & 0   &  \cdots & 0
\end{pmatrix}  }

\newcommand{\detmatrix}[4]{ %%%removed 3rd column and 2nd row.
\begin{vmatrix}
  \lambda - #1 & -1 & 0  & \cdots & 0 \\
 - #2 & \lambda  & -1 &\cdots & 0 \\
\vdots & \vdots & \ddots & \ddots & \vdots \\
 - #3 & 0   & 0 & \cdots & -1 \\
 - #4 & 0   & 0 & \cdots & \lambda
\end{vmatrix} }

\begin{document}

\title[Scaling zeta functions and recursive structure of lattice strings]{Multifractal Analysis via Scaling Zeta Functions and Recursive Structure of Lattice Strings}

\author[R.~de~Santiago]{Rolando~de~Santiago}
\address{Department of Mathematics and Statistics\\ California State Polytechnic University\\ Pomona, California 91768 USA}
%\curraddr{Department of Mathematics and Statistics,
%Case Western Reserve University, Cleveland, Ohio 43403}
\email{rdesantiago@csupomona.edu}

%	Information for second author
\author[M.~L.~Lapidus]{Michel~L.~Lapidus}
\address{Department of Mathematics\\ 
University of California\\ Riverside, California 92521-0135 USA}
\email{lapidus@math.ucr.edu}
\thanks{The work of the second author (M.~L.~Lapidus) was partially supported by the US National Science Foundation under the research grant DMS--1107750, as well as by the Institut des Hautes Etudes Scientifiques (IHES) where the second author was a visiting professor in the Spring of 2012 while this paper was written.}

%    Information for third author
\author[S.~A.~Roby]{Scott~A.~Roby}
\address{Department of Mathematics\\ 
University of California\\ Riverside, California 92521-0135 USA}
\email{roby@math.ucr.edu}

%    Information for fourth author
\author[J.~A.~Rock]{John~A.~Rock}
\address{Department of Mathematics and Statistics\\ California State Polytechnic University\\ Pomona, California 91768 USA}
\email{jarock@csupomona.edu}

%    General info
\date{July 21, 2012}



\begin{abstract}
The multifractal structure underlying a self-similar measure stems directly from the weighted self-similar system (or weighted iterated function system) which is used to construct the measure. This follows much in the way that the dimension of a self-similar set, be it the Hausdorff, Minkowski, or similarity dimension, is determined by the scaling ratios of the corresponding self-similar system via Moran's theorem. The multifractal structure allows for our definition of scaling regularity and $\alpha$-scaling zeta functions motivated by the geometric zeta functions of \cite{LapvF6} and, in particular, the partition zeta functions of \cite{ELMR,LapRo1}. Some of the results of this paper consolidate and partially extend the results of \cite{ELMR,LapRo1} to a multifractal analysis of certain self-similar measures supported on compact subsets of a Euclidean space. Specifically, the $\alpha$-scaling zeta functions generalize the partition zeta functions of \cite{ELMR,LapRo1} when the choice of the family of partitions is given by the natural family of partitions determined by the self-similar system in question. Moreover, in certain cases, self-similar measures can be shown to exhibit lattice or nonlattice structure with respect to specified scaling regularity values in a sense which extends that of \cite{LapvF6}. Additionally, in the context provided by generalized fractal strings viewed as measures, we define generalized self-similar strings, allowing for the examination of many of the  results presented here in a specific overarching context and for a connection to the results regarding the corresponding complex dimensions as roots of Dirichlet polynomials in \cite{LapvF6}. Furthermore, generalized lattice strings and recursive strings are defined and shown to be very closely related.
\end{abstract}

\maketitle

\section{Introduction and summary}
\label{sec:IntroductionAndSummary}

The abscissae of convergence of the $\alpha$-\emph{scaling zeta functions} associated with a self-similar measure provide a multifractal spectrum akin to the classic symbolic multifractal spectrum, called the \emph{scaling multifractal spectrum}. The technique described herein allows for partial, yet extensive, generalizations of the main results on the abscissae of convergence of partition zeta functions found in \cite{LapRo1} and especially \cite{ELMR}. In a variety of cases, the Hausdorff dimensions of Besicovitch subsets of a self-similar set are recovered. In other cases, the classic symbolic multifractal spectrum of a self-similar measure is recovered. Moran's theorem is fully recovered in the special case where the measure in question (supported on a given self-similar set) is the natural mass distribution associated with a specific probability distribution. 

Along with the development of $\alpha$-scaling zeta functions, a generalized setting for self-similar and lattice generalized fractal strings is developed and \emph{recursive strings} are introduced in this paper. Indeed, \emph{generalized self-similar strings} provide a framework in which one can analyze certain special cases of $\alpha$-scaling zeta functions. It is also shown that \emph{generalized lattice strings} and \emph{recursive strings} are intimately related.

In terms of applications, multifractal analysis is the study of a variety of mathematical, physical, dynamical, probabilistic, statistical, and biological phenomena from which families of fractals may arise. Such objects and behaviors are often modeled by mass distributions, or measures, with highly irregular and intricate structure. These multifractal measures, or simply multifractals, stem from phenomena such as weather, structure of lightning, turbulence, distribution of galaxies, spatial distribution of oil and minerals, cellular growth, internet traffic, and financial time series. See \cite{BMP,CM,ELMR,Falc,Fed,Hut,LVM,BM2,MauUrb,Ol1,Ol2,Ol3,Ol4,STZ}.

The structure of this paper is summarized as follows:

Section \ref{sec:MultifractalAnalysisOfSelfSimilarSystems} provides a summary of results on classical approaches to multifractal analysis of self-similar measures established in the literature which are most pertinent to the results presented in this paper. In particular, the manner in which words are associated to the structure of self-similar measures (i.e., the way in which symbolic dynamics is employed) is discussed. See \cite{Bes1,BMP,CM,DekLi,Egg,Falc,Fed,Hut,Jaf,LapvF5,MauUrb,Mor,Ol1,Ol2,Ol3,Ol4,STZ} for classical and known results on self-similar sets and multifractal analysis of self-similar measures and other multifractal objects.

In Section \ref{sec:FractalStringsAndComplexDimensions}, definitions and results regarding (generalized) fractal strings and complex dimensions from \cite{LapvF6} are recalled. Further, \emph{generalized self-similar strings} are defined and are shown to have, by design, complex dimensions that are completely determined by the roots of Dirichlet polynomials, as examined (thoroughly) in \cite[Chs.~2 \& 3]{LapvF6}. That is, generalized self-similar strings provide a context in which the self-similar structures considered throughout the paper can be related to the study of Dirichlet polynomials performed in \cite{LapvF6}. In addition to \cite{LapvF6}, see \cite{ELMR,LLVR,LapPeWin,LapPeWi2,LapRo1,LapvF5,LVM,Rock} for further notions and uses of fractal strings and complex dimensions in a variety of contexts.

In Section \ref{sec:GeneralizedLatticeStringsAndLinearRecurrenceRelations}, \emph{generalized lattice strings} and \emph{recursive strings} are defined and the intimate connections between them are discussed. In particular, it is shown that every generalized lattice string is a recursive string and exhibits properties which are determined by a naturally corresponding linear recurrence relation. It is also shown that the complex dimensions of a recursive string are the complex dimensions of a naturally corresponding generalized lattice string. See \cite{deSa} for a more detailed development of the connections between generalized lattice strings, linear recurrence relations, and recursive strings.

In Section \ref{sec:MultifractalAnalysisViaScalingRegularityAndScalingZetaFunctions}, scaling regularity is used to define families of fractal strings associated with a given self-similar measure, giving rise to the definition of $\alpha$-\emph{scaling zeta functions} and the notion of multifractal spectrum as the abscissae of convergence of these zeta functions. This technique is motivated by and partially extends the results on partition zeta functions found in \cite{ELMR,LapvF6,LapRo1,Rock}.

Finally, Section \ref{sec:FurtherResultsAndFutureWork} provides preliminary investigations of some further problems which expand upon the results of this paper. In particular, the $\alpha$-scaling zeta functions of certain self-similar measures are shown to actually be hypergeometric series. This relationship is central to the material studied in \cite{EssLap}.  See \cite{BatErd} for more information on hypergeometric series. Also, a family of self-similar measures which do not satisfy any of the conditions of the theorems and corollaries in Section \ref{sec:MultifractalAnalysisViaScalingRegularityAndScalingZetaFunctions} is investigated, motivating further research.

\section{Multifractal analysis of self-similar systems}
\label{sec:MultifractalAnalysisOfSelfSimilarSystems}

Multifractal analysis of a measure $\nu$ concerns the fractal geometry of objects such as the sets $E_t$ of points $x\in E$ for which the measure $\nu(B(x,r))$ of the closed ball $B(x,r)$ with center $x$ and radius $r$ satisfies
\begin{align*}
\lim_{r \rightarrow 0^+}\frac{\log{\nu(B(x,r))}}{\log{r}} &= t,
\end{align*}
where $t\geq0$ is the {\it local H\"older regularity} and $E$ is the support of $\nu$. Roughly speaking, multifractal analysis is the study of the ways in which a Borel measure behaves locally like $r^t$.

\subsection{Multifractal spectra}
\label{sec:MultifractalSpectra}

The multifractal spectra of Definitions \ref{def:GeometricHausdorffMultifractalSpectrum} and \ref{def:SymbolicHausdorffMultifractalSpectrum} along with Proposition \ref{prop:ClassicMultifractalSpectra} below are presented as found in \cite{Ol4}, as well as the corresponding references therein. See especially the work of Cawley and Mauldin in \cite{CM}.

\begin{definition}
\label{def:GeometricHausdorffMultifractalSpectrum}
The \emph{geometric Hausdorff multifractal spectrum} $f_g$ of a Borel measure $\nu$ supported on $E$ is given by
\begin{align*}
%\label{eqn:GeometricHausdorffMultifractalSpectrum}
f_g(t)	&:=\dim_H(E_t), 
\end{align*}
where $t\geq0$, $\dim_H$ is the Hausdorff dimension, and
\begin{align*}
%\label{eqn:GeometricHausdorffMultifractalSpectrumSubset}
E_t	&:=\left\{x \in E : \lim_{r \rightarrow 0^+}\frac{\log{\nu(B(x,r))}}{\log{r}} =t\right\}. 
\end{align*}
\end{definition}

\subsection{Self-similar systems}
\label{sec:SelfSimilarSystems}

Self-similar systems lie at the heart of many of the results presented in this paper.

\begin{definition}
\label{def:SelfSimilarSystem}
Given $N \in \N, N \geq 2$, a \emph{self-similar system} $\bfPhi=\{\Phi_j\}_{j=1}^N$ is a finite family of contracting similarities on a complete metric space $(X,d_X)$. Thus, for all $x,y \in X$ and each $j=1,\ldots,N$ we have 
\begin{align*}
d_X(\Phi_j(x),\Phi_j(y))&=r_jd_X(x,y),
\end{align*}
where $0<r_j<1$ is the \emph{scaling ratio} (or Lipschitz constant) of $\Phi_j$ for each $j=1,\ldots,N$. 

The \emph{attractor} of $\bfPhi$ is the nonempty compact set $F \subset X$ defined as the unique fixed point of the contraction mapping 
\begin{align}
\label{eqn:SelfSimilarSystemAsMap}
\displaystyle
\bfPhi(\cdot)	& := \bigcup_{j=1}^N \Phi_j (\cdot)
\end{align}
on the space of compact subsets of $X$ equipped with the Hausdorff metric. That is, $F=\bfPhi(F)$. The set $F$ is also called the \emph{self-similar set} associated with $\bfPhi$.

A self-similar system (or set) is \emph{lattice} if there is a unique real number $r$ and positive integers $k_j$ such that $0<r<1$ and $r_j=r^{k_j}$ for each $j=1,\ldots,N$. Otherwise, the self-similar system (or set) is \emph{nonlattice}. 
\end{definition}
 
\begin{remark}
\label{rmk:OSCThroughout}
For clarity of exposition, only self-similar systems on some Euclidean space $\R^d$  $(d \in \N)$, with $X \subset \R^d$, are considered.\footnote{Throughout this paper, $ \N $ denotes the set of positive intergers: $\N=\lbrace 1,2,3,\ldots \rbrace$.} Furthermore, throughout the paper we consider only self-similar systems which satisfy the \emph{open set condition}. (See \cite{CM,ELMR,Falc,Hut,LapvF6}.) Recall that a self-similar system $\bfPhi$ satisfies the open set condition if there is a nonempty open set $V \subset \R^d$ such that $\bfPhi(V) \subset V$ and $\Phi_j(V) \cap \Phi_k(V)= \emptyset$ for each $j, k \in \{1,\ldots,N\}$ where $j \neq k$. Results presented in \cite{Ol4} and \cite{STZ}, for example, specifically do \emph{not} require the open set condition to be satisfied.
\end{remark}

The notion of self-similar ordinary fractal strings which are either lattice or nonlattice as defined below follows from \cite[Ch.~2]{LapvF6}. These notions are extended and generalized in various ways throughout this paper.

\begin{definition}
\label{def:SelfSimilarLatticeOrdinaryFractalStrings}
Let $\bfPhi$ be a self-similar system on $\R$ such that $\sum_{j=1}^Nr_j<1$  and satisfying the open set condition on a compact interval $I$. Denote the endpoints of $I$ by $a_1$ and $a_2$, and assume that	 there are $j_1,j_2 \in \{1,\ldots,N\}$ such that $a_1 \in \Phi_{j_1}(I)$ and $a_2 \in \Phi_{j_2}(I)$. The complement of the attractor $F$, given by $I\setminus F$, is a \emph{self-similar ordinary fractal string}. Let $I^o$ denote the interior of $I$. The lengths of the connected components of $I^o\setminus \bfPhi(I)$, called the \emph{gaps} of $\bfPhi$, are denoted by $\bfg=(g_1,\ldots,g_K)\in(0,\infty)^K$ where $K \in \N$. If, additionally, there is a unique unique real number $r$ and positive integers $k_j$ such that $0<r<1$ and $r_j=r^{k_j}$ for each $j=1,\ldots,N$, then $I \setminus F$ is \emph{lattice}. Otherwise, $I \setminus F$ is \emph{nonlattice}. 
\end{definition}

The following notation, which is motivated by the notation of self-similar (ordinary) fractal strings in \cite[Ch.~2]{LapvF6}, allows for a clearer comparison between the main results herein and the classical results found in, for instance, \cite{CM,Falc,Hut}.

\begin{notation}
\label{not:SelfSimilarityWords}
For each $k \in \N\cup\{0\}$, let $\scrJ_k=\{1,\ldots,N\}^k$ denote the set of all finite sequences of length $k$ in the symbols $\{1,\ldots,N\}$ (i.e., words). For $k=0$, let $\scrJ_0$ be the set consisting of the empty word. Let $\scrJ=\cup_{k=0}^{\infty}\scrJ_k$; hence, $\scrJ$ is the set of all finite sequences (or words) in the symbols $\{1,\ldots,N\}$. Let $\scrJ_\infty$ denote the set of all one-sided infinite sequences (or words) in the symbols $\{1,\ldots,N\}$. For $J \in \scrJ$, let $|J|$ denote the number of components (i.e., the length) of $J$ and define $|J|=\infty $ if $J \in \scrJ_\infty$. For a word $J$ (in either $\scrJ$ or $\scrJ_\infty$) and each $n \in \N$, $n \leq |J|$, let $J|n$ denote the truncation of $J$ at its $n$th component and let $J|0$ denote the empty word. More specifically, $J|n=j_1j_2\ldots j_n$ if $J$ begins with the letters $j_1,j_2,\ldots,j_n$.  For $J\in\scrJ$, define the contracting similarity $\Phi_J$ by
\begin{align*}
\Phi_J	&:=\Phi_{j_{|J|}} \circ \Phi_{j_{|J|-1}} \circ \cdots \circ \Phi_{j_1}. 
\end{align*}
The scaling ratio of $\Phi_J$ is given by $r_J=\prod_{q=1}^{|J|}r_{\pi_q(J)}$, where $\pi_q(\cdot)$ is the projection of a word onto its $q$th component. For the empty word $J|0$, let $\Phi_{J|0}$ denote the identity map and let $r_{J|0}=1$. For a set $E \subset X$ and a word $J$, let
\begin{align*}
\displaystyle
E_J	&:= \Phi_J(E).
\end{align*}
Finally, define $\tau : \scrJ_\infty \rightarrow \R^d$ by $\{\tau(J)\}  := \cap_{n=0}^{\infty} E_{J|n}$.
\end{notation}

In Theorem 9.1 of \cite{Falc}, for instance, it is shown via the Contraction Mapping Principle that \eqref{eqn:SelfSimilarSystemAsMap} uniquely defines the attractor $F$ as the fixed point of the map $\bfPhi(\cdot)$. Moreover, for any compact non-empty set $E$ such that $\Phi_j(E) \subset E$ for each $j=1,\ldots,N$, we have 
\begin{align*}
F 	&= \bigcap_{n=0}^{\infty} \bigcup_{|J|=n} E_{J},
\end{align*}
where the notation ``$|J|=n$'' indicates that, for each $n \in \N \cup \{0\}$, the corresponding union runs over all $J \in \scrJ$ such that $|J|=n$ (i.e., over all words of length $n$).

For the support $F$, its Hausdorff dimension $\dim_{H}(F)$ is given by the unique positive real solution $D$ of the Moran equation \eqref{eqn:Moran}; see \cite{Falc,Hut,Mor}. Equivalently, $D$ is equal to the Minkowski dimension of $F$, denoted by $\dim_{M}(F)$. In \cite{BM2}, $D$ is called the `similarity dimension' (or `exponent') of $F$.

\begin{theorem}[Moran's Theorem]
\label{thm:MoransTheorem}
Let $\bfPhi$ be a self-similar system with scaling ratios $\{r_j\}_{j=1}^N$ that satisfies the open set condition. Then the Hausdorff \emph{(}and Minkowski\emph{)} dimension of the attractor $F$ is given by the unique \emph{(}and hence, positive\emph{)} real solution $D$ of the equation  
\begin{align}
\label{eqn:Moran}
\sum_{j=1}^N r_j^\sigma	&= 1, \quad \sigma \in \R.
\end{align}
\end{theorem}

\begin{remark}
\label{rmk:ProofOfMoransTheorem}
Moran's original result in \cite{Mor} is provided in dimension one, but it extends to any ambient dimension $d \geq 1$. Furthermore, the proof of Moran's Theorem, as presented in \cite{Falc} and \cite{Hut} for instance, makes use of the mass distribution principle (see mass distribution principle 4.2 in \cite{Falc}) and the ``natural mass distribution'' $\mu$. Recall that $\mu$ is the  self-similar measure determined, as described in the next section, by the self-similar system $\bfPhi$ and the probability vector  $\bfp=(r_1^D,\ldots,r_N^D)$. Note that $\mu$ is supported on the attractor (or self-similar set) $F$.
\end{remark}

\subsection{Self-similar measures and scaling regularity}
\label{sec:SelfSimilarityAndScalingRegularity}

The multifractal measures in the context of this paper are constructed as follows (cf.~\cite{CM,ELMR,Hut}). 

Let $\bfPhi=\{\Phi_j\}_{j=1}^N$ be a self-similar system that satisfies the open set condition with scaling ratios $\bfr=(r_1,\ldots,r_N)$, where $0<r_j<1$ for each $j=1,\ldots,N$. Let $\bfp=(p_1,\ldots,p_N)$ be a probability vector such that $0 \leq p_j \leq 1$ (hence $\sum_{j=1}^Np_j=1$). A self-similar measure $\mu$ supported on the attractor $F$ of the self-similar system $\bfPhi$ can be constructed via the probability vector $\bfp$ and the equation
\begin{align}
\label{eqn:MeasureSelfSimilarity}
\mu(E)	&= \sum_{j=1}^N p_j \cdot \mu(\Phi_j^{-1}(E)),
\end{align}
which holds for all compact subsets $E$ of $\R^d$. The self-similar measure $\mu$ is uniquely defined as the unique fixed point of the contraction implied by \eqref{eqn:MeasureSelfSimilarity} on the space of regular Borel measures with unit total mass equipped with the $L$-metric (see \cite{Hut}). In this setting, we refer to the pair $(\bfPhi,\bfp)$ as a {\it weighted self-similar system}. When $\bfp=(r_1^D,\ldots,r_N^D)$, the resulting self-similar measure $\mu$ is called the \emph{natural Hausdorff measure} (or \emph{natural mass distribution}) associated with the attractor $F$.

It is worth noting that the analysis of a given self-similar measure $\mu$  developed below depends only on the scaling ratios $\bfr$ of the corresponding self-similar system $\bfPhi$ and the probability distribution determined by $\bfp$.

\begin{notation}
\label{not:ProductsOfScalingRatios}
For each $J \in \scrJ$, let 
\begin{align*}
%\label{eqn:ProductsOfScalingRatios}
r_J=\prod_{q=1}^{|J|}r_{\pi_q(J)} \quad \textnormal{and} \quad
p_J=\prod_{q=1}^{|J|}p_{\pi_q(J)}. 
\end{align*}
Thus, for $J \in \scrJ_\infty$ and each $k \in \N$ we have $r_{J|k}=r_{j_1} \cdots r_{j_k}$ and $p_{J|k}=p_{j_1} \cdots p_{j_k}$.
\end{notation}

\begin{remark}
From our perspective, there is a specific reason to distinguish between the elements of $\scrJ_\infty$ and those of $\scrJ$. Indeed, the classical symbolic Hausdorff multifractal spectrum of Definition \ref{def:SymbolicHausdorffMultifractalSpectrum} below is defined in terms of the truncation of the elements of $\scrJ_\infty$ whereas our results, found in Section \ref{sec:MultifractalAnalysisViaScalingRegularityAndScalingZetaFunctions}, are stated in terms of elements of $\scrJ$.
\end{remark}

\begin{definition}
\label{def:ScalingRegularity}
Let $J \in \scrJ$. The \emph{scaling regularity} of $J$ is the value $A_{\bfr,\bfp}(J)$ given by 
\begin{align*}
A_{\bfr,\bfp}(J)	&:=\log_{r_J}{p_J} = \frac{\log{p_J}}{\log{r_J}},
\end{align*}
where $r_J$ and $p_J$ are defined in Notation \ref{not:ProductsOfScalingRatios}. Alternately, $A_{\bfr,\bfp}(J)$ is the unique real number $\alpha$ defined by $r_J^\alpha = p_J$.
\end{definition}

\begin{remark}
\label{rmk:ScalingAndCoarseRegularity}
Note that the scaling regularity $A_{\bfr,\bfp}$ depends only on the scaling ratios $\bfr$ of $\bfPhi$ and the probability vector $\bfp$ but not on the contracting similarities $\Phi_j \in \bfPhi$. Indeed, the results presented in this paper are independent of the contracting similarities themselves. However, we do require, as mentioned above, that a self-similar system $\bfPhi$ satisfies the open set condition. Also note that for each $J \in \scrJ$, the scaling regularity of $J$ coincides with the coarse H\"{o}lder regularity of $E_J$; see \cite{CM} and \cite{ELMR}, for instance.
\end{remark}

\subsection{The symbolic Hausdorff multifractal spectrum}
\label{sec:SymbolicHausdorffMultifractalSpectrum}

Self-similar measures are often called \emph{multifractal measures} since, as will be discussed, whenever a self-similar measure is \emph{not} the natural Hausdorff measure of the support, it decomposes the support into an amagalmation of fractal sets.

\begin{definition}
\label{def:SymbolicHausdorffMultifractalSpectrum}
Let $\mu$ be the self-similar measure determined by a weighted self-similar system $(\bfPhi,\bfp)$. The \emph{symbolic Hausdorff multifractal spectrum} $f_s$ of $\mu$ is given by
\begin{align*}
f_s(t)	&:=\dim_H\left\{\tau(J) : J \in \scrJ_\infty \textnormal{ and } \lim_{n \rightarrow \infty} A_{\bfr,\bfp}(J|n) =t \right\} %\label{eqn:SymbolicHausdorffMultifractalSpectrum}
\end{align*}
for $t\geq0$, where the map $\tau$ is defined at the very end of Notation \ref{not:SelfSimilarityWords} and the number $A_{\bfr,\bfp}(J)$ is given in Definition \ref{def:ScalingRegularity}. Here and henceforth, given $E \subset \R^d$, $\dim_H(E)$ denotes the Hausdorff dimension of $E$.
\end{definition}

The following proposition is a simplified version of a similar proposition in \cite{CM}.

\begin{proposition}
\label{prop:ClassicMultifractalSpectra}
Let $\mu$ be the unique self-similar measure on $\R^d$ defined by a weighted self-similar system $(\bfPhi,\bfp)$ which satisfies the open set condition. Then
\begin{align*}
%\label{eqn:ClassicMultifractalSpectra}
f_g(t) = f_s(t), \qquad t\geq 0. 
\end{align*}
\end{proposition}

Due to Proposition \ref{prop:ClassicMultifractalSpectra} and the fact that scaling regularity plays a central role in Section \ref{sec:MultifractalAnalysisViaScalingRegularityAndScalingZetaFunctions}, for a self-similar measure $\mu$, focus is put on the symbolic Hausdorff multifractal spectrum $f_s$ in the remainder of this section. A more complete development of the properties of $f_s$ described here can be found in \cite[\S 1]{CM}.

\begin{remark}
\label{rmk:StructureOfMultifractalSpectrum}
The self-similar measure $\mu$ uniquely defined by a weighted self-similar system $(\bfPhi,\bfp)$ attains maximum and minimum scaling regularity values $\alpha_{\min}$ and $\alpha_{\max}$ which, in turn, define the compact interval on which the sets $E_t$ from Definition \ref{def:GeometricHausdorffMultifractalSpectrum} are nonempty. Define $\alpha_j:=\log_{r_j}{p_j}$ for each $j \in \{1,\ldots,N\}$. These extreme scaling regularity values are given by
\begin{align*}
%\label{eqn:ExtremeScalingRegularityValues}
\alpha_{\min} &= \min\left\{\alpha_j : j \in \{1,\ldots,N\} \right\}, \qquad \alpha_{\max} = \max\left\{\alpha_j : j \in \{1,\ldots,N\} \right\}.
\end{align*}
%The domain of $f_s$ is then $[\alpha_{\min},\alpha_{\max}]$.

If $\bfp=(r_1^D,\ldots,r_N^D)$, where $D$ is the Hausdorff dimension of the attractor $F$ of the self-similar system $\bfPhi$ (that is, if $\mu$ is the natural Hausdorff measure of its support, which is the attractor $F$), then $\alpha_{\min}=\alpha_{\max}=D$ and the domain of $f_s$ is the singleton $\{D\}$. In general, the domain of $f_s$ is the collection of all nonnegative real values $t$ such that $A_{\bfr,\bfp}(J|n) =t$. Hence, the domain of $f_s$ is $[\alpha_{\min},\alpha_{\max}]$.

In the case where the domain of $f_s$ is a non-degenerate interval $[\alpha_{\min},\alpha_{\max}]$, we have that $f_s$ is concave and $f_s'$ is unbounded near $\alpha_{\min}$ and $\alpha_{\max}$. Moreover, the unique value $t_1$ such that $t_1=f_s(t_1)$ is the \emph{information dimension} of $\mu$, and 
\begin{align*}
%\label{eqn:MaxOfSpectrumIsHausdorffDimension}
\max\{ f_s(t) : t \in [\alpha_{\min},\alpha_{\max}] \}	&= \dim_H(F) = \dim_M(F),
\end{align*}
where $\dim_M(F)$ denotes the Minkowski (or box) dimension of $F$; see \cite[Chs.~2 \& 3]{Falc} for the definition of Hausdorff and Minkowski dimension. In the context of an \emph{ordinary fractal string} $\Omega$, we are also concerned with the \emph{inner} Minkowski dimension of its boundary $\partial\Omega$; see \cite[\S 1.1]{LapvF6} and Definition \ref{def:MinkowskiDimensionAndMeasurability} below.
\end{remark}

The following example is also studied in the context of partition zeta functions; see \cite{LapRo1}, \cite[\S 5.2]{ELMR}, and \cite{Rock}. 

\begin{example}[Measures on the Cantor set]
\label{eg:MeasuresOnTheCantorSet}
The Cantor set, denoted $C$, is the unique nonempty attractor of the lattice self-similar system $\bfPhi_C$ on $[0,1]$ given by the two contracting similarities $\varphi_1(x)=x/3$ and $\varphi_2(x)=x/3+2/3$ with scaling ratios $\bfr=(1/3,1,3)$. The Hausdorff dimension, and equivalently the Minkowski dimension, of $C$ is the unique real solution of the corresponding Moran equation (cf. \eqref{eqn:Moran}): $2\cdot3^{-s}=1$. Thus, $\dim_{H}C=\dim_{M}C=\log_{3}2=:D_C$. 

When $\bfPhi_C$ is weighted by $\bfp=(1/2,1/2)=(1/3^{D_C},1/3^{D_C})$, the corresponding self-similar measure is the natural mass distribution (i.e., the natural Hausdorff measure) $\mu_C$ of the Cantor set. Such measures are used to find lower bounds on the Hausdorff dimension of their supports; see \cite[Ch.~9]{Falc} and Remark \ref{rmk:NaturalHausdorffMeasures} below.

When $\bfPhi_C$ is weighted by $\bfp=(1/3,2/3)$, we obtain the self-similar measure $\beta$, called the \emph{binomial measure}, which exhibits the following properties: $\alpha_{\min}=1-\log_{3}2$; $\alpha_{\max}=1$; and for $t \in [1-\log_{3}2,1]$, the geometric (and symbolic) Hausdorff multifractal spectrum is given by
\begin{align*}
f_g(t)	&= f_s(t)= -\left(\frac{1-t}{\log_{3}2}\right)\log_{3}\left(\frac{1-t}{\log_{3}2}\right) -\left(1-\frac{1-t}{\log_{3}2}\right)\log_{3}\left(1-\frac{1-t}{\log_{3}2}\right).
\end{align*}
See \cite[\S 5.2]{ELMR} and \cite{Rock} for details.

Note that, in the case of $\mu_C$, the only scaling regularity value attained by any corresponding word $J$ is $A_{\bfr,\bfp}(J)=\log_{3}2$. However, in the case of $\beta$, the attained scaling regularity values depend on $\bfk=(k_1,k_2)$ where $\sum\bfk:=k_1+k_2=|J|$, $k_1$ denotes the number of times 1 appears in $J$, and $k_2$ denotes the number of times 2 appears in $J$. Moreover, for a fixed $n \in \N$ with $\sum\bfk=n$, we have 
\begin{align*}
\#\left\{J \in \scrJ_{n} : \#\left\{q: \pi_q(J)=j \right\} = k_j, j\in \left\{1,2\right\} \right\}		&=\frac{n!}{k_1!k_2!}=:\binom{\sum\bfk}{\bfk},
\end{align*}
where, in general, $\binom{\sum\bfk}{\bfk}$ denotes the multinomial coefficient and the projection $\pi_q$ is defined in Notation \ref{not:SelfSimilarityWords}. 

This decomposition of the words associated with a given weighted self-similar system 
via scaling regularity along with the corresponding multinomial coefficients lies at the heart of the approach to multifractal analysis developed in Section \ref{sec:MultifractalAnalysisViaScalingRegularityAndScalingZetaFunctions} below. 
\end{example}

The multifractal spectrum of a self-similar measure $\mu$ developed in Section \ref{sec:MultifractalAnalysisViaScalingRegularityAndScalingZetaFunctions} of this paper is determined by the abscissae of convergence of the $\alpha$-scaling zeta functions. The motivation for this approach, and the analogous approach taken in \cite{LapRo1} and \cite{ELMR}, is a classic theorem of Besicovitch and Taylor and its significant extension which is at the heart of the theory of complex dimensions of fractal strings developed in \cite{LapvF6}. (See Theorem 1.10 of \cite{LapvF6} along with Theorem \ref{thm:BesicovitchAndTaylor} below.)

\subsection{Besicovitch subsets of the attractor of a self-similar system}
\label{sec:BesicovitchSubsetsOfAttractor}

A probability vector can be used not only to define a self-similar measure supported on a self-similar set, but also to decompose the support of such a measure into a family of disjoint subsets with interesting fractal properties of their own.

\begin{definition}
\label{def:BesicovitchSubsets}
Let $F$ be the attractor of a self-similar system $\bfPhi$ and let $\bfq=(q_1,\ldots,q_N)$ be a probability vector. The \emph{Besicovitch subset} $F(\bfq) \subset F$ is defined as follows:
\begin{align*}
%\label{eqn:BesicovitchSubset}
F(\bfq) &:= \left\{ x \in F : \lim_{k\rightarrow \infty} \frac{\#_j(x|_k)}{k} = q_j, \hs j \in \{1,\ldots,N\} \right\},
\end{align*}
where $x|_k$ is the truncation of $x$ at its $k$-th term in the expansion implied by $\bfPhi$ via $\tau(J)=x$ and $\#_j(x|_k)$ is the number of times the term $j$ appears in $x|_k$ (really, the number of times $j$ appears in $J|k$).
\end{definition}

\begin{remark}
\label{rmk:DensityOfBesicovitchSubsets}
A little thought shows that, for a probability vector $\bfq$, the Besicovitch subset $F(\bfq)$ is dense in the support $F$. That is, $\overline{F(\bfq)}=F$, where $\overline{F(\bfq)}$ is the closure of $F(\bfq)$. Furthermore, if the Minkowski dimension $\dim_{M}(F)$ exists (i.e., if the upper and lower Minkowski (or box) dimensions of $F$ coincide), then Proposition 3.4 of \cite{Falc} implies that $\dim_{M}(F(\bfq))=\dim_{M}(F)$. (Actually, the upper and lower Minkowski dimensions of a set are always equal, respectively, to those of its closure.) Throughout this paper, either the Minkowski or Hausdorff dimension of a given set will be considered, depending on the context.
\end{remark}

Care needs to be taken in the above definition. Some $x \in F$ may have more than one $J \in \scrJ_\infty$ where $x=\tau(J)$: however, this has no effect on the following proposition regarding the Hausdorff dimension of a Besicovitch subset $F(\bfq)$. See \cite{CM} for a proof of the following proposition, and see \cite{Bes1} and \cite{Egg} for related classical results.

\begin{proposition}
\label{prop:HausdorffDimensionsOfBesicovitchSubsets}
Let $F(\bfq)$ be the Besicovitch subset of the attractor $F$ of a self-similar system $\bfPhi$ with scaling ratios $\bfr=(r_1,\ldots,r_N)$ determined by a probability vector $\bfq$. Then
\begin{align*}
%\label{eqn:HausdorffDimensionOfBesicovitchSubsets}
\dim_H(F(\bfq)) &= \frac{\sum_{j=1}^N q_j \log{q_j}}{\sum_{j=1}^N q_j \log{r_j}}.
\end{align*}
\end{proposition}

\section{Fractal strings and complex dimensions}
\label{sec:FractalStringsAndComplexDimensions} 

The material found in this section provides a brief summary of pertinent results from the theory of complex dimensions of fractal strings developed by Lapidus and van Frankenhuijsen in \cite{LapvF6}.

\subsection{Generalized fractal strings}
\label{sec:GeneralizedFractalStrings}

For a (local) measure $\eta$ on $(0,\infty)$, denote the total variation of $\eta$ by $|\eta|$. For a bounded measurable set $S$ we have,\footnote{A \emph{local measure} $\eta$ on $(0,\infty)$ is a $\C$-valued function on the Borel $\sigma$-algebra of $(0,\infty)$ whose restriction to any bounded subinterval is a complex measure. If $\eta$ is $[0,\infty]$-valued, then $\eta$ is simply a locally bounded positive measure on $(0,\infty)$  and it is called a \emph{local positive measure}.}
\begin{align*}
%\label{eqn:TotalVariation}
|\eta|(S)=\sup\left\{\sum_{k=1}^{m}|\eta(S_k)|\right\},
\end{align*}
where $m \in \N$ and $\{S_k\}_{k=1}^m$ ranges over all finite partitions of $S$ into disjoint measurable subsets of $(0,\infty)$. Recall that $|\eta|=\eta$ if $\eta$ is positive and that $|\eta|$ is a positive measure.

\begin{definition}
\label{def:GeneralizedFractalString}
A \emph{generalized fractal string} is either a local complex or a local positive measure $\eta$ on $(0,\infty)$ which is supported on a subset of $(x_0,\infty)$ for some $x_0>0$. The \emph{dimension} of $\eta$, denoted $D_\eta$, is the \emph{abscissa of convergence} of the Dirichlet integral $\zeta_{|\eta|}(\sigma)=\int_0^\infty x^{-\sigma}|\eta|(dx)$. That is,
\begin{align*}
%\label{eqn:GeneralizedDimension}
D_\eta	&:=\inf\left\{\sigma \in \R \,|\, \int_0^\infty x^{-\sigma}|\eta|(dx) < \infty \right\}.
\end{align*} 
The \emph{geometric zeta function} of $\eta$ is the Mellin transform of $\eta$ given by
\begin{align*}
%\label{eqn:GeneralizedGeometricZetaFunction}
\zeta_\eta(s)	&=\int_0^\infty x^{-s}\eta(dx),
\end{align*}	
for $\textnormal{Re}(s)>D_\eta$. Let $W \subset \C$ be a window\footnote{As in \cite{LapvF6}, we are interested in the meromorphic extension of the geometric zeta function $\zeta_\calL$ to suitable regions. To this end, consider the \emph{screen} $S$ as the contour 
\begin{align*}
%\label{eqn:Screen}
S:S(t)+it \quad &(t \in \R),
\end{align*}
where $S(t)$ is a Lipschitz continuous function $S:\R \to [-\infty,D_\calL]$. Also, consider the \emph{window} $W$ as the closed set 
\begin{align*}
%\label{eqn:Window}
W	&= \{s \in \C : \textnormal{Re}(s) \geq S(\textnormal{Im}(s))\}
\end{align*}
and assume that $\zeta_\eta$ has a meromorphic continuation to an open connected neighborhood of $W$ satisfying suitable polynomial growth conditions (as in \cite[\S 5.3]{LapvF6}).} on an open neighborhood of which $\zeta_\eta$ has a meromorphic extension. By a mild abuse of notation, both the geometric zeta function of $\eta$ and its meromorphic extension are denoted by $\zeta_\eta$. The set of (\emph{visible}) \emph{complex dimensions} of $\eta$, denoted by $\calD_\eta$, is given by
\begin{align*}
%\label{eqn:ComplexDimensionsOfADiscreteGeneralizedFractalString}
\calD_\eta	&= \left\{ \omega \in W : \zeta_\eta \textnormal{ has a pole at } \omega \right\}.
\end{align*}
\end{definition}
In the case where $\zeta_\eta$ has a meromorphic extension to $W=\C$, the set $\calD_\eta$ is referred to as the \emph{complex dimensions} of $\eta$.

Generalized fractal strings have two notable predecessors: fractal strings and ordinary fractal strings.

\begin{definition}
\label{def:FractalString}
A \emph{fractal string} $\calL=\{\ell_j\}_{j=1}^\infty$ is a nonincreasing sequence of positive real numbers which tend to zero. Hence, $\lim_{j \rightarrow \infty}\ell_j=0$.
\end{definition}

\begin{remark}
\label{rmk:FractalStringClassicSetting}
As in \cite{LapPeWin}, but unlike in the classic geometric setting of \cite{LapvF6}, we do not require $\sum_{j=1}^\infty\ell_j < \infty$. That is, in \cite{LapvF6} an \emph{ordinary fractal string} is a bounded open subset $\Omega$ of $\R$ and $\calL$ denotes the sequence of lengths of the disjoint open intervals whose union is $\Omega$; see, e.g., \cite{LapPo1}, \cite[Chs.~1 \& 2]{LapvF6}, and \cite{LapPeWin}. We note, however, that in \cite{LapPo2}, \cite[Chs. 3 \& 10]{LapvF6}, and \cite{LapPeWin,LapPeWi2}, for example, the underlying sequence of scales is allowed to satisfy $\sum_{j=1}^\infty\ell_j=\infty$. 
\end{remark}

There is a natural relationship between discrete generalized fractal strings $\eta$ and fractal strings $\calL$. Recall that the Dirac mass at $x \in (0,\infty)$, denoted by $\delta_{\lbrace x\rbrace}$, is the measure given by
\begin{align*}
%\label{eqn:DiracMass}
\delta_{\lbrace x \rbrace}(S) &:= 
\begin{cases}
1,& x \in S, \\
0,& x \notin S,
\end{cases}
\end{align*} 
where $S \subseteq \R$ (for example). So, a fractal string 
\begin{align*}
\calL	&=\{\ell_j\}_{j=1}^{\infty}=\{l_n\,|\, l_n\textnormal{ has multiplicity }m_n,n \in \N\}
\end{align*} 
defines the generalized fractal string $\eta$ as follows:
\begin{align*}
\eta		&= \sum_{j=1}^{\infty}\delta_{\{\ell_j^{-1}\}}=\sum_{n=1}^{\infty}m_n\delta_{\lbrace l_n^{-1}\rbrace}.
\end{align*}
For such $\eta$, it immediately follows that $\zeta_\calL = \zeta_\eta$ and $\calD_\calL=\calD_\eta$. 

The following theorem, which is Theorem 1.10 of \cite{LapvF6}, is a restatement of a classical theorem of Besicovitch and Taylor (see \cite{BesTa}) formulated in terms of ordinary fractal strings, as first observed in \cite{Lap2}. (A direct proof can be found in \cite{LapvF6}, \emph{loc.~cit.}) For the definition of (inner) Minkowski dimension as used below, see \cite[\S 1.1]{LapvF6} and Definition \ref{def:MinkowskiDimensionAndMeasurability} below.

\begin{theorem}
\label{thm:BesicovitchAndTaylor}
Suppose $\Omega$ is an ordinary fractal string with infinitely many lengths denoted by $\calL$. Then the abscissa of convergence of $\zeta_\calL$ coincides with the \emph{(}inner\emph{)}Minkowski dimension of $\partial\Omega$. That is, $D_\calL=\dim_M(\partial\Omega)$, where $\dim_M(\partial\Omega)$ denotes the \emph{(}inner\emph{)} Minkowski dimension of $\partial\Omega$.\footnote{Unlike in \cite{BesTa}, the fact that the inner Minkowski Dimension is used in \cite{LapvF6} allows for not requiring any additional assumptions about $\partial\Omega$ or about $\Omega$.}
\end{theorem}

In terms of such meromorphic extensions, a given geometric zeta function encountered throughout this paper falls into one of two categories: (i) the geometric zeta function has a meromorphic extension to all of $\C$, or (ii) the geometric zeta function is similar to a hypergeometric series and, hence, in general, the appropriate extension is yet to be determined (see Section \ref{sec:FurtherResultsAndFutureWork} and \cite{BatErd,EssLap}). For a self-similar ordinary fractal string (see Definition \ref{def:SelfSimilarLatticeOrdinaryFractalStrings}), the geometric zeta function has a closed form which allows for a meromorphic extension to all of  $\C$.  This closed form is given in the following theorem, which is Theorem 2.3 in \cite{LapvF6}. Compare this theorem with Theorem \ref{thm:ClosedFormOfGeometricZetaFunction} below, which is a completely analogous result regarding the closed form of the geometric zeta function of a \emph{generalized self-similar string}.

\begin{theorem}
\label{thm:GeometricZetaFunctionOrdinarySelfSimilarString}
Let $\Omega$ be a self-similar ordinary fractal string with lengths $\calL$. Then the geometric zeta function $\zeta_\calL$ has a meromorphic continuation to the whole complex plane, given by
\begin{align}
\label{eqn:GeometricZetaFunctionOrdinarySelfSimilarString}
\zeta_\calL(s)	&=\frac{L^s\sum_{k=1}^Kg_k^s}{1-\sum_{j=1}^Nr_j^s}, \quad s \in \C.
\end{align}
Here, $L=\zeta_\calL(1)$ is the total length of $\Omega$.
\end{theorem}

\begin{example}[The Fibonacci string]
\label{eg:TheFibonacciString}

%\begin{figure}
%\includegraphics[scale=.75]{fibonaccistringconstruction.eps}
%\caption{Constructing the Fibonacci string $\calL_{\textnormal{Fib}}$.}
%\label{fig:ConstructFibonacciString}
%\end{figure}

Consider the lattice self-similar system on the interval $[0,4]$ given by $\Phi_1(x)=x/2$ and $\Phi_2(x)=x/4+3$ (i.e., $r_1=1/2$ and $r_2=1/4=1/2^2$). This self-similar system generates an attractor $F$ and a lattice ordinary fractal string $\Omega=[0,4] \setminus F$ whose lengths are given by the Fibonacci string $\calL_{\textnormal{Fib}}$. (See \cite[\S2.3.2]{LapvF6}.) The Fibonacci string $\calL_{\textnormal{Fib}}$ is the fractal string with distinct lengths given by $l_n=2^{-n}$ occurring with multiplicity $m_n=F_n$, where $F_n$ is the $n$th Fibonacci number and $n \in \N \cup\{0\}$. (Hence, $F_0=F_1=1,F_2=2,\ldots$.)  Via Theorem \ref{thm:GeometricZetaFunctionOrdinarySelfSimilarString}, the geometric zeta function of $\calL_{\textnormal{Fib}}$ is given by
\begin{align}
\label{eqn:GeometricZetaFunctionFibonacci}
\zeta_{\textnormal{Fib}}(s)	&:= \zeta_{\calL_{\textnormal{Fib}}}(s) =\sum_{n=0}^{\infty}F_{n}2^{-ns}=\frac{1}{1-2^{-s}-4^{-s}} %\label{eqn:gzffib}
\end{align}
for $s \in \C$, and the dimension $D_\textnormal{Fib}$ is the unique real-valued solution of the equation 
\begin{align}
\label{eqn:DFibSolution}
2^{-2s}+2^{-s}	&=1, \quad s \in \C. 
\end{align}
Moreover, the complex dimensions of $\calL_{\textnormal{Fib}}$ are the complex roots of \eqref{eqn:DFibSolution}.\footnote{These roots are obtained by solving the quadratic equation $z^2+z-1=0$ with $z=2^{-s}, s \in \C$.} Thus, we have
\begin{align}
\label{eqn:FibonacciComplexDimensions}
\calD_{\textnormal{Fib}}	&=\left\{D_{\textnormal{Fib}}+izp : z \in \Z \right\} \cup \left\{-D_{\textnormal{Fib}}+i(z+1/2)p : z \in \Z \right\},
\end{align}
where $\varphi=(1+\sqrt{5})/2$ is the Golden Ratio, $D_{\textnormal{Fib}}=\log_2{\varphi}$, and the oscillatory period is $p = 2\pi/\log{2}$.

Accordingly, the ordinary self-similar fractal string $[0,4] \setminus F$ is lattice in the sense of Definition \ref{def:SelfSimilarLatticeOrdinaryFractalStrings}. Furthermore, as they are given by the Fibonacci numbers, the multiplicities $m_n=F_n$ satisfy the linear recurrence relation for $n \geq 2$ given by
\begin{align}
\label{eqn:FibonacciLinearRecurrenceRelation}
F_n	&= F_{n-1}+F_{n-2},
\end{align}
with initial conditions $\bfs_\textnormal{Fib}:=(F_0,F_1)=(1,1)$.
\end{example}

Connections between generalized lattice strings and linear recurrence relations are examined in Section \ref{sec:GeneralizedLatticeStringsAndLinearRecurrenceRelations}.

\subsection{Minkowski measurability and lattice/nonlattice dichotomy}
\label{sec:MinkowskiMeasurabilityAndLatticeNonlatticeDichotomy}

The following theorem is a partial restatement of Theorem 8.15 of \cite{LapvF6} that provides a criterion for the Minkowski measurability of an ordinary fractal string $\Omega$ satisfying certain mild restrictions. First, we introduce a few useful tools. (See \cite[\S 1.1]{LapvF6} for more information such as detailed definitions of lower and upper Minkowski contents for ordinary fractal strings.)

\begin{definition}
\label{def:MinkowskiDimensionAndMeasurability}
Let $V(\varepsilon)$ be the $1$-dimensional Lebesgue measure of the \emph{inner tubular neighborhood} of $\Omega$ given by the set $\{x\in\Omega:d(x,\partial\Omega)<\varepsilon\}$. The \emph{inner
Minkowski dimension} of $\partial\Omega$, denoted $D=D_\calL$, is given by
\[ 
D= \inf\{t\geq0 :V(\varepsilon)=O(\varepsilon^{t}) \textnormal{ as } \varepsilon \to 0^+\}.
\]
The boundary $\partial\Omega$ of an ordinary fractal string $\Omega$ is \emph{Minkowski measurable} if the limit $\lim_{\varepsilon \rightarrow 0^+}V(\varepsilon)\varepsilon^{D-1}$ exists in $(0,\infty)$. 
\end{definition}

In the following theorem, the equivalence of statements (ii) and (iii) below for an arbitrary ordinary fractal string is established in \cite{LapPo1} (without any conditions on $\calL$ other than $D_\calL\neq 0,1$).

\begin{theorem}[Criterion for Minkowski measurability]
\label{thm:CriterionForMinkowskiMeasurability}
Let $\Omega$ be an ordinary fractal string whose geometric zeta function $\zeta_\calL$ has a meromorphic extension which satisfies certain mild growth conditions.\footnote{Specifically, Theorem \ref{thm:CriterionForMinkowskiMeasurability} holds if $\zeta_\calL$ is \emph{languid} (see \cite[Def.~5.2]{LapvF6}) for a screen passing between the vertical line $\textnormal{Re}(s)=D_\calL$ and all the complex dimensions of (the corresponding fractal string) $\calL$ with real part strictly less than $D$, and not passing through 0.} Then the following statements are
equivalent\emph{:}
\begin{enumerate}
	\item $D$ is the only complex dimension with real part $D$, and it is simple.
	%\item $N_\calL(x)=cx^D+o(x^D)$ as $x \to \infty$ for some positive constant $c$.
	\item $\partial\Omega$ is Minkowski measurable.
	\item $\ell_j = Lj^{-1/D}(1+o(1))$ as $j\to\infty$, for some $L>0$.
\end{enumerate}
%Moreover, if any of these conditions is satisfied, then the Minkowski content $\calM$ of $\partial\Omega$ is given by
%\begin{align}
%\label{eqn:MinkowskiContentFormula}
%\calM	&= \frac{c2^{1-D}}{1-D} = 2^{1-D}\frac{\textnormal{res}(\zeta_\calL(s);D)}{D(1-D)}. 
%\end{align}
\end{theorem}

Theorem \ref{thm:CriterionForMinkowskiMeasurability} applies to all lattice self-similar ordinary fractal strings. Specifically, if $\Omega$ is a lattice string, then the mild growth conditions are satisfied by $\zeta_\calL$ and 
%by Theorem \ref{thm:StructureOfZerosOfDirichletPolynomials} 
there are no complex dimensions other than $D$ which have real part $D$, so the boundary of $\Omega$ is not Minkowski measurable. On the other hand, Theorem \ref{thm:CriterionForMinkowskiMeasurability} does not apply to all nonlattice self-similar strings since there are some for which $\zeta_\calL$ does not satisfy the growth conditions for a screen of the type described in footnote 6; see \cite[Example~5.32]{LapvF6}. Nonetheless, we have the following theorem which partially summarizes Theorems 8.23 and 8.36 of \cite{LapvF6}.

\begin{theorem}[Lattice/nonlattice dichotomy]
\label{thm:LatticeNonlatticeDichotomy}
A self-similar ordinary fractal string $\Omega$ is nonlattice if and only if its boundary $\partial\Omega$ is Minkowski measurable. 
\end{theorem}

\begin{remark}
An extension of a part of Theorem \ref{thm:LatticeNonlatticeDichotomy} for suitable classes of self-similar tilings (and sets or systems) of higher-dimensional Euclidean spaces is provided in \cite{LapPeWi2}, using results of \cite[Ch.~8]{LapvF6} and \cite{LapPeWin}. (See also the relevant references therein.) Furthermore, an interesting study of a nonlinear analogue of Theorem \ref{thm:LatticeNonlatticeDichotomy} (and related counter-examples for certain self-conformal sets) in the real line is conducted in \cite{KesKom}.
\end{remark}

In the next section, we summarize (in an extended self-similar setting) some key results on the generalized fractal strings (i.e., fractal strings viewed as measures) of Lapidus and van Frankenhuijsen found in \cite[Ch.~4]{LapvF6}.

\subsection{Generalized self-similar strings}
\label{sec:GeneralizedSelfSimilarStrings}

In this section, a particular form of generalized fractal string, called \emph{generalized self-similar string}, is defined. The corresponding complex dimensions, by design, are given by the roots of a naturally associated Dirichlet polynomial equation; see \cite[Ch.~3]{LapvF6}. The geometric zeta function of a generalized self-similar string has a meromorphic continuation established in Theorem \ref{thm:ClosedFormOfGeometricZetaFunction} which allows for the determination of the complex dimensions of the recursive strings defined in Section \ref{sec:GeneralizedLatticeStringsAndLinearRecurrenceRelations} and the complex dimensions of certain scaling zeta functions in Section \ref{sec:MultifractalAnalysisViaScalingRegularityAndScalingZetaFunctions}.

\begin{definition}
\label{def:GeneralizedSelfSimilarString}
A discrete generalized fractal string $\eta$ is \emph{self-similar}, and may be referred to as a \emph{generalized self-similar string}, if there are $K,N \in \N$ such that for some
$\bfg=(g_1,\ldots,g_K) \in (0,\infty)^K, \bfr=(r_1,\ldots,r_N) \in (0,1)^N,$ and $\bfm=(m_1,\ldots,m_N) \in \C^N$, we have
\begin{align*}
\eta	&= \sum_{k=1}^K \sum_{J \in \scrJ} m_J \delta_{\{g_k^{-1} r_J^{-1}\}},
\end{align*}
where $r_J=\prod_{q=1}^{|J|}r_{\pi_q(J)}$ and $m_J$ is defined in an identical fashion. For the empty word $J|0$, we let $m_{J|0}=r_{J|0}=1$. We refer to components of the vectors $\bfg,\bfr$, and $\bfm$ (or sometimes the vectors themselves) as the \emph{gaps, initial scaling ratios,} and \emph{initial multiplicities} of $\eta$, respectively.
\end{definition}

The following theorem, which is an immediate consequence of the results of \cite[\S 3.3]{LapvF6}, but will be useful in the sequel, determines a closed form of the geometric zeta function of a generalized self-similar string (cf.~Theorem \ref{thm:GeometricZetaFunctionOrdinarySelfSimilarString} above and Theorem 2.3 and Equation (3.21) in \cite{LapvF6}). This closed form allows for the meromorphic continuation of the geometric zeta function to all of the complex plane and, therefore, an extension of the theory of complex dimensions for self-similar strings of \cite{LapvF6}. When the multiplicities $m_j$ are all integral and positive, this fact, along with the detailed study (conducted in \cite[Ch.~3]{LapvF6}) of the periodic or almost periodic distribution of the complex dimensions in the lattice or nonlattice case, respectively, was used in an essential manner in the work of \cite{LapPeWin} on tube formulas for higher-dimensional self-similar sets and tilings. (See also the earlier papers by the first two authors of \cite{LapPeWin} quoted therein, along with \cite{LapPeWi2}.) Moreover, the proof of the following theorem is included so that one may compare and contrast with the results on $\alpha$-scaling functions presented in Section \ref{sec:MultifractalAnalysisViaScalingRegularityAndScalingZetaFunctions} below, especially the decomposition of the corresponding multiplicities.

\begin{theorem}
\label{thm:ClosedFormOfGeometricZetaFunction}
Let $\eta$ be a generalized self-similar string. Then the geometric zeta function of $\eta$ has a meromorphic continuation to $\C$ given by
\begin{align*}
\zeta_\eta(s)	&=\frac{\sum_{k=1}^K g_k^s}{1-\sum_{j=1}^Nm_jr_j^s}, \quad \textnormal{for } s \in \C.
\end{align*} 
\end{theorem}

\begin{proof}
For $q=0$, $J|0$ is the empty word and we have $m_{J|0}=r_{J|0}=1$. For each $q \in \N$ we have 
\begin{align*}
\sum_{J \in \scrJ : |J|=q } m_Jr_J^s	&= \sum_{\nu_1=1}^N \cdots \sum_{\nu_q=1}^N m_{\nu_1}r_{\nu_1}^s \cdots m_{\nu_q}r_{\nu_q}^s
= \left(\sum_{j=1}^Nm_jr_j^s\right)^q.
\end{align*}
Now, for $\textnormal{Re}(s)>D_\eta$, $\zeta_\eta$ is given by 
\begin{align*}
\zeta_\eta(s) = \int_0^\infty x^{-s}\eta(dx) &= \sum_{k=1}^K \sum_{J \in \scrJ} m_J(g_kr_J)^s
= \sum_{k=1}^K g_k^s \left( \sum_{q=0}^\infty \left(\sum_{j=1}^N m_jr_j^s \right)^q \right).
\end{align*}
Note that since $\textnormal{Re}(s)>D_\eta$, the series converges because $|\sum_{j=1}^Nm_jr_j^s|<1$. Moreover, we have 
\begin{align*}
\zeta_\eta(s)	&= \frac{\sum_{k=1}^K g_k^s}{1-\sum_{j=1}^Nm_jr_j^s}.
\end{align*}
Hence, by the Principle of Analytic Continuation, $\zeta_\eta$ has a meromorphic continuation to all of $\C$ given by the last formula. 
\end{proof}

\begin{example}[Self-similar ordinary fractal strings]
\label{eg:SelfSimilarOrdinaryFractalStrings}
The geometric zeta function of a self-similar ordinary fractal string is given by the geometric zeta function of an appropriately defined generalized self-similar string. Basically, an ordinary fractal string is self-similar if it is the complement of a self-similar set with respect to a certain type of closed interval. (See \cite[Ch.~2]{LapvF6}.) For instance, the geometric zeta function $\zeta_{\textnormal{Fib}}$ of the Fibonacci string given in \eqref{eqn:GeometricZetaFunctionFibonacci} is the geometric zeta function of the generalized self-similar string determined by $K=g_1=1$, $\bfr=(1/2,1/4)$, and $\bfm=(1,1)$. (See Example \ref{eg:TheFibonacciString}.) Another such example is the Cantor string $\Omega_{CS}$ which is the complement in $[0,1]$ of the classic Cantor set. (See \cite{ELMR} and \cite{LapvF6} for further information on the Cantor string.)
%The lengths of the Cantor string are given by the fractal string 
%\begin{align}
%\label{eqn:CantorStringLengths}
%\calL_{CS}	&:=\{3^{-n}: 3^{-n} \textnormal{ has multiplicity } m_n=2^{n-1}, n \in %\N\}.
%\end{align}
%The geometric zeta function of the Cantor string, denoted $\zeta_{CS}$, is given by
%\begin{align}
%\label{eqn:CantorStringGeometricZetaFunction}
%\zeta_{CS}(s)	&:= \sum_{n=1}^\infty 2^{n-1}3^{-ns} = \frac{3^{-s}}{1-2\cdot3^{-s}}.
%\end{align}
%The generalized self-similar string which generates $\zeta_{CS}$ as it geometric zeta function is determined by $N=K=1$, $g_1=r_1=1/3$, and $m_1=2$. 

The self-similar systems which generate the Fibonacci string and the Cantor string, respectively, are lattice. Hence, as described in Section \ref{sec:GeneralizedLatticeStringsAndLinearRecurrenceRelations}, the corresponding generalized self-similar strings are \emph{recursive strings} (see Definition \ref{def:RecursiveString}). We note that in the case of the Fibonacci string, this fact is foreshadowed in the linear recurrence relation \eqref{eqn:FibonacciLinearRecurrenceRelation}.
\end{example}

\begin{example}[Generalized Cantor strings]
\label{eg:GeneralizedCantorStrings}
A \emph{generalized Cantor string} $\mu$ is a generalized self-similar string of the form
\begin{align*}
%\label{eqn:GeneralizedCantorStrings}
\mu=\sum_{n=0}^\infty b^n\delta_{\lbrace r^{-n}\rbrace},
\end{align*}
where $0<r<1$ and $b>0$. (Note that $b$ is not required to be an integer.) That is, as a generalized self-similar string, $\mu$ is determined by $N=K=g_1=1$, $\bfr=(r)$, and $\bfm=(b)$. Generalized Cantor strings are studied in \cite[Chs.~8 \& 10]{LapvF6}, where it is shown (among many other things) that the geometric zeta function of a generalized Cantor string has a meromorphic extension to all of $\C$ given by 
\begin{align}
\label{eqn:GeometricZetaFunctionOfGeneralizedCantorString}
\zeta_\mu(s)	&=\frac{1}{1-b\cdot r^{s}}, \quad s \in \C.
\end{align}

In the special case where $b=2^{-1}$ and $r=3^{-1}$, one obtains a generalized Cantor string which can not be realized geometrically as an ordinary fractal string (since $b=m_1=2^{-1}$ is nonintegral). The geometric zeta function of $\mu$, after meromorphic extension to $\C$, is given by 
\begin{align*}
\zeta_\mu(s)	&=\frac{1}{1-2^{-1}\cdot 3^{-s}}, \quad s \in \C.
\end{align*}
Note that the ``dimension'' $D_\mu=-\log_{3}2$ is \emph{negative}.\footnote{By ``dimension''here we mean the abscissa of convergence of $\zeta_\mu$, as in Definition \ref{def:GeneralizedFractalString}.} Indeed, by allowing the ``multiplicities'' $\bfm\neq\mathbf{0}$ to comprise complex numbers in Definition \ref{def:GeneralizedSelfSimilarString}, one is able to study self-similar structures with respect to measures which do not necessarily (or rather, do not readily) correspond to geometric objects.
\end{example}

We conclude this section with an indication of how the framework of generalized self-similar strings ties to material elsewhere in the literature. Specifically, the \emph{scaling measure} and \emph{scaling zeta function} of a self-similar system, as defined below, are studied in \cite{LapPeWin} and \cite{LapvF6}. In the setting of this paper, a scaling measure is a generalized self-similar string where $m_j=K=g_1=1$ for each $j=1,\ldots,N$, and the scaling zeta function is the corresponding geometric zeta function. As mentioned in Section \ref{sec:MinkowskiMeasurabilityAndLatticeNonlatticeDichotomy} above, there are deep connections between the Minkowski measurability of the attractor of a self-similar system and the structure of the complex dimensions of corresponding zeta functions such as the scaling zeta function. (See \cite[Ch.~2 \& \S 8.4]{LapvF6}; see also \cite{LapPeWin,LapPeWi2}.)

\begin{definition}
\label{def:ScalingMeasure}
For a self-similar system $\bfPhi$, the \emph{scaling measure} is the associated generalized fractal string given by
\begin{align*}
%\label{eqn:ScalingMeasure}
\eta_\bfPhi	&:=\sum_{J \in \scrJ} \delta_{\lbrace r_J^{-1}\rbrace}.
\end{align*}
The \emph{scaling zeta function} of $\bfPhi$ is given by
\begin{align*}
%\label{eqn:ScalingZetaFunction}
\zeta_\bfPhi(s)	&:= \frac{1}{1-\sum_{j=1}^Nr_j^s}, \quad s \in \C.
\end{align*}
According to Definition \ref{def:GeneralizedFractalString}, it is just the geometric zeta function (i.e., the Mellin transform) of the generalized fractal string $\eta_\bfPhi$. 
\end{definition}

\begin{remark}
\label{rmk:ScalingMeasureGatherMultiplicities}
%As mentioned at the beginning of this section, we do not require the scaling ratios $r_j$ to be distinct. This is also the setting for Theorem \ref{thm:MoransTheorem} (Moran's Theorem) and our results in Section \ref{sec:MultifractalAnalysisViaScalingRegularityAndScalingZetaFunctions}, but it is not the setting for the results pertaining to Dirichlet polynomials in Section \ref{sec:GeneralizedSelfSimilarityAndRootsOfDirichletPolynomials} . If one wishes to take advantage of the results in Section \ref{sec:GeneralizedSelfSimilarityAndRootsOfDirichletPolynomials} in this current setting, one may simply take the sum of the multiplicities of a repeated scale as the multiplicity for that distinct scale, then repeat for all such distinct scales. 
Every scaling measure $\eta_\bfPhi$ defines a generalized self-similar string and if the self-similar system $\bfPhi$ is lattice, then $\eta_\bfPhi$ is a generalized lattice string (see Definition \ref{def:GeneralizedLatticeString} below). %Indeed, given a self-similar system $\bfPhi$, the scaling ratios $\bfr=(r_1,\dots,r_N)$ of $\bfPhi$ comprise $N_0 \in \N$ distinct scaling ratios $\bft=(t_1,\ldots,t_{N_0})$ with multiplicities $m_q \in \N$, $q=1,\ldots,N_0$ with $N_0 \leq N$. Hence, consider the Dirichlet polynomial associated with $\bfPhi$ given by
%\begin{align}
%\label{eqn:AssociatedDirichletPolynomial}
%f_\bfPhi(s)	&:= 1-\sum_{q=1}^{N_0}m_qt_q^s.
%\end{align} 
%Consequently, by Theorem \ref{thm:ClosedFormOfGeometricZetaFunction} the meromorphic continuation of the geometric zeta function of $\eta_\bfPhi$ to all of $\C$ is given by the scaling zeta function $\zeta_\bfPhi$ in \eqref{eqn:ScalingZetaFunction}. By Corollary \ref{cor:ComplexDimensionsOfAGeneralizedSelfSimilarString}, the complex dimensions of $\eta_\bfPhi$ are the complex roots of the equation $f_\bfPhi(s)=0$. Also, the unique real root of the equation $f_\bfPhi(s)=0$, the abscissa of convergence of the geometric zeta function of $\eta_\bfPhi$, and the Hausdorff and Minkowski dimensions of the attractor $F$ all coincide.
\end{remark}

\section{Generalized lattice strings and linear recurrence relations}
\label{sec:GeneralizedLatticeStringsAndLinearRecurrenceRelations}

In this section, we discuss yet another notion of lattice structure. This time, it pertains to generalized self-similar strings. In particular, the results of this section extend accordingly to self-similar sets which are subsets of some Euclidean space (not just the real line) and certain cases of scaling zeta functions of self-similar measures as discussed in Section \ref{sec:MultifractalAnalysisViaScalingRegularityAndScalingZetaFunctions}. 

Linear recurrence relations are also shown, in this section, to be intimately related to the \emph{generalized lattice strings} defined here.

\begin{definition}
\label{def:GeneralizedLatticeString}
A generalized self-similar string $\eta$ is \emph{lattice} if there is a unique $0<r<1$, called the \emph{multiplicative generator} of $\eta$, and positive integers $k_j$ such that $r_j=r^{k_j}$ for each $j=1,\ldots,N$. A lattice generalized self-similar string may also be referred to as a \emph{generalized lattice string}. If a generalized self-similar string is not lattice, it is called a \emph{generalized nonlattice self-similar string}.
\end{definition}

The study of generalized self-similar strings (and hence, of generalized lattice and nonlattice strings), as well as of the structure of their complex dimensions, is the object of \cite[Ch.~3]{LapvF6}.

\begin{remark}
Every lattice ordinary fractal string (see Definition \ref{def:SelfSimilarLatticeOrdinaryFractalStrings} and the Fibonacci and Cantor strings in Examples \ref{eg:TheFibonacciString} and \ref{eg:SelfSimilarOrdinaryFractalStrings}) and every generalized Cantor string (see Example \ref{eg:GeneralizedCantorStrings}) can be realized as a generalized lattice string. 
\end{remark}

What follows is a discussion of a connection between linear recurrence relations and the (possibly complex) multiplicities stemming from generalized lattice strings and lattice ordinary fractal strings.

\subsection{Linear recurrence relations}
\label{sec:LinearRecurrenceRelations}

A brief summary of some relevant material on linear recurrence relations is provided in this section. For a more detailed introduction to recurrence relations, see \cite{Bal}.

\begin{definition}
\label{def:LinearRecurrenceRelation}
A sequence $\{s_n\}_{n=0}^{\infty}\subset\C$ satisfies a \emph{linear recurrence relation} $\mathcal{R}$ if there exist $d\in\N$ and $\bfa\in\C^d$ with $\pi_{d}(\bfa)\neq0 $ such that for all $n\geq d$, we have
\begin{align*}
 s_n=a_1s_{n-1}+\cdots+a_ds_{n-d}.
\end{align*}
The positive integer $d$ is called the \emph{degree} of the linear recursion. Furthermore, the constant vector $\mathbf{a}:=(a_1,\ldots,a_d)$ is called the \emph{kernel} of $\mathcal{R}$. The \emph{characteristic equation} of $\Rec$ is given by 
\begin{align} 
\label{eqn:CharacteristicEquation}
\recpoly{\varphi}{d}, \quad \varphi \in \C.
\end{align}

For a given sequence $\{s_n\}_{n=0}^{\infty}$ which satisfies the linear recurrence relation $\Rec$, the first $d$ terms of the sequence $\{s_n\}_{n=0}^{\infty}$ are called the \emph{initial conditions}, and they are denoted by the vector $\bfs:=(s_0,\ldots,s_{d-1})$. 
\end{definition}

\begin{remark}
\label{rmk:DeterminingLinearRecurrenceRelations}
Each linear recurrence relation $\Rec$ is completely determined by its kernel $\bfa$. Also, each sequence $\{s_n\}_{n=0}^\infty$ which satisfies a linear recurrence relation $\Rec$ is completely determined by the corresponding kernel $\bfa$ and the initial conditions $\bfs:=(s_0,\ldots,s_{d-1})$. 
\end{remark}

Some of the properties of linear recurrence relations can be understood and analyzed in the context of linear algebra. A brief synopsis is provided below; see \cite{Bal} and \cite{deSa} for more information.

\begin{definition}
\label{def:KernelMatrixAndNthSequenceMatrix}
Suppose $\{s_n\}_{n=1}^{\infty}$ satisfies a linear recurrence relation $\Rec$. The \emph{kernel matrix} $A$ and the \emph{$n$th sequence matrix} $S_n$ are respectively given by
\begin{align*}
A	&:=\cfmatrix{a_1}{a_2}{a_{d-1}}{a_{d}}, \quad S_n:=\recmatrix{s_{n}}{s_{1+n}}{s_{2+n}}{s_{d-2+n}}{s_{d-1+n}}{s_{d+n}}{s_{2d-4+n}}{s_{2d-3+n}}{s_{2d-2+n}}.
\end{align*}
\end{definition}

\begin{theorem} 
\label{thm:MatrixEquationForLinearRecurrence}
Suppose $\{s_n\}_{n=0}^{\infty}$ satisfies a linear recurrence relation $\Rec$. Then $S_0A^n=S_n$.
\end{theorem}

\begin{proof}
For each $j \in \N$ we have $S_jA = S_{j+1}$. Thus, for each $n \in \N$ we have
\begin{align*}
S_0A^n &= S_1A^{n-1} =\cdots = S_n. 
\end{align*}
\end{proof}

\begin{proposition}
\label{prop:EigenvalueIsRoot}
A complex number $\lambda$ is an eigenvalue of the kernel matrix $A$ of a recurrence relation $\Rec$ if and only if $\lambda$ is a solution of the characteristic equation of $\Rec$ given by \eqref{eqn:CharacteristicEquation}. 
\end{proposition}

The proof is omitted as it follows immediately from the definitions and some linear algebra.

\begin{remark}
\label{rmk:RecursionInFGCDZF}
Recursion relations in the context of measures with another type of self-similarity property are studied in \cite[\S 4.4.1]{LapvF6}. There, such measures are allowed to have mass near zero and are assumed to be absolutely continuous with respect to $dx/x$, the Haar measure on the multiplicative group $\R_+^*$. These results are compared and contrasted with the results of this paper in \cite{deSa}. 
\end{remark}

In the next section, linear recurrence relations extend and are related to generalized lattice strings.

\subsection{Recursive strings}
\label{sec:RecursiveStrings}

In this section, we present a new type of generalized fractal string, called a \emph{recursive string}. Recursive strings are closely related to generalized lattice strings as described in this section, especially via Theorems \ref{thm:LatticeAndRecurrence} and \ref{thm:RecursiveStringZetaFunction}. See \cite{deSa} for a more detailed analysis of recursive strings, linear recurrence relations, and connections to generalized lattice strings. 

\begin{definition}
\label{def:RecursiveString}
Given $K\in\N$, $0<r<1$, $\bfg=(g_1,\ldots,g_K) \in (0,1)^K$, and a sequence $\{s_n\}_{n=0}^\infty \subset \C$ satisfying a linear recurrence relation $\calR$, the \emph{recursive string} $\eta_\Rec (\bfs;\cdot)$ is the discrete generalized fractal string given by 
\begin{align*}
\eta_\calR(\bfs;\cdot)	&=\sum_{k=1}^K\sum_{n=0}^{\infty} s_n\delta_{\lbrace g_k^{-1} r^{-n}\rbrace}(\cdot).
\end{align*} 
The real number $r$ is called the \emph{multiplicative generator} of $\eta_\Rec (\bfs;\cdot)$ and the components of the vector $\bfg$ (or sometimes $\bfg$ itself) are(is) called the \emph{gaps} of $\eta_\Rec (\bfs;\cdot)$. 
\end{definition}

The geometric zeta function, dimension, and complex dimensions of a recursive string $\eta_\Rec (\bfs;\cdot)$ are denoted, respectively, by 
\[
\zeta_\calR(\bfs;\cdot):=\zeta_{\eta_\Rec (\bfs;\cdot)}, \quad D_\calR(\bfs):=D_{\eta_\Rec (\bfs;\cdot)}, \quad \textnormal{and} \quad \calD_\calR(\bfs):=\calD_{\eta_\Rec (\bfs;\cdot)}.
\]

\begin{remark}
\label{rmk:DeterminingRecursiveStrings}
Since every linear recurrence relation $\calR$ is completely determined by its kernel $\bfa$ and every sequence $\{s_n\}_{n=0}^\infty \subset \C$ which satisfies $\calR$ is further determined by the initial conditions $\bfs$, we have that every recursive string $\eta_\Rec (\bfs;\cdot)$ is completely determined by the kernel $\bfa$, initial conditions $\bfs$, gaps $\bfg$, and multiplicative generator $r$.
\end{remark}

\begin{theorem}
\label{thm:LatticeAndRecurrence}
Every generalized lattice string $\eta$ is a recursive string. That is,
\begin{align*}
%\label{eqn:LatticeAndRecurrence}
\eta	&= \sum_{n=0}^\infty s_n \delta_{\lbrace r^{-n}\rbrace},
\end{align*}
where $0<r<1$ and the sequence of multiplicities $\{s_n\}_{n=0}^\infty \subset \C$ satisfies a linear recurrence relation $\Rec$.
\end{theorem}

\begin{proof}
Since $\eta$ is lattice, there exist a self-similar system $\bfPhi$ with scaling ratios $\bfr=(r_1,\ldots,r_N)$ and a unique $0<r<1$ such that $r_j = r^{k_j}$ for some positive integer $k_j$ and each $j=1,\ldots,N$. Let $k=\max\{k_j : j \in \{1,\ldots,N\} \}$. For each $q=1,\ldots,k$, define $a_q$ to be the sum of the (complex) multiplicities of $r^q$ with respect to the self-similar system. That is, $a_q:=\sum_{j:r_j=r^q}m_j$. Let $s_0=1$ and for $n \in \N$, let $s_n$ be the multiplicity of the length $r^n$. Since $r^n=r^qr^{n-q}$ for $1\leq q\leq k$, each instance of $r^q$ will contribute $s_{n-q}$ to the total multiplicity of the length $r^n$. Specifically, for all $n \geq k$,
\begin{align*}
s_n	&= \sum_{q=1}^k a_qs_{n-q} = a_1s_{n-1} + \cdots + a_ks_{n-k}.
\end{align*}
This is the desired linear recurrence relation $\mathcal{R}$.
\end{proof}

The following corollary provides a well-known fact regarding the Hausdorff dimension and Minkowski dimension of a lattice self-similar system. A simple proof is provided in part to illuminate the deep connections between generalized lattice strings and linear recurrence relations.

\begin{corollary}
\label{cor:DimensionFromLinearRecurrenceRelation}
Suppose $\bfPhi$ is a lattice self-similar system with attractor $F$. Then 
there is a sequence of positive integer multiplicities $\{s_n\}_{n=0}^\infty$ which satisfies the linear recurrence relation $\Rec$ corresponding to $\bfPhi$ such that the scaling measure $\eta_\bfPhi$ satisfies
\begin{align*}
%\label{eqn:LatticeAndRecurrenceCorollary}
\eta_\bfPhi	&= \sum_{n=0}^\infty s_n \delta_{\lbrace r^{-n} \rbrace}.
\end{align*}
Moreover, the Hausdorff dimension and Minkowski dimension of $F$ are given by
\begin{align*}
%\label{eqn:DimensionFromLinearRecurrenceRelation}
\dim_H(F)	&=\dim_M(F)	= -\log_r{\varphi}, 
\end{align*}
where $\varphi$ is the unique positive root the of characteristic equation of $\Rec$. 
\end{corollary}

\begin{proof}
By Theorem \ref{thm:MoransTheorem}, $\dim_H(F)$ and $\dim_M(F)$ are given by the unique real solution $D$ of \eqref{eqn:Moran}. 
Since $\bfPhi$ is lattice, $r_j = r^{k_j}$, for $ 0< r <1$ and $k_j \in \N$.  Letting $m=\max\{k_j:j=1,\ldots,N\}$, $a_j$ be the multiplicity of $r^j$, and $x=r^{-s}$,
\eqref{eqn:Moran} becomes
\begin{align*}
1 	&= \sum_{j=0}^ma_jr^{js} = a_1x^{-1} + \cdots + a_mx^{-m}.
\end{align*}
Multiplying through by $x^m$ yields 
\begin{align*}
x^m	&= a_1x^{m-1} + \cdots + a_m,
\end{align*}
which is a polynomial in $x$. Since the coefficients are all positive, Descarte's Rule of Signs implies the existence of a unique positive root $\varphi$.
Thus, $\varphi = r^{-D}$, which implies $D = -\log_r{\varphi}$.
\end{proof}

The following theorem and its corollaries show that a recursive string is nearly a generalized lattice string. In particular, the set of complex dimensions of a recursive string is contained in or equal to the set of complex dimensions of a naturally related generalized lattice string.

\begin{theorem}
\label{thm:RecursiveStringZetaFunction}
Let $\calR$ be a linear recurrence relation of degree $d$. Let $\eta_\Rec (\bfs;\cdot)$ be the recursive string determined by the kernel $\bfa$ of $\calR$, initial conditions $\bfs$, gaps $\bfg$, and multiplicative generator $0<r<1$. Then 
\begin{align}
\label{eqn:RecursiveStringZetaFunction}
\zeta_\Rec(\bfs;s) &=g(s)\cdot \frac{\displaystyle\sum_{n=0}^{d-1}s_nr^{ns} - \sum_{l=1}^{d-1}a_lr^{ls}\sum_{n=0}^{d-1-l}s_nr^{ns}}
{\displaystyle 1-\sum_{j=1}^da_jr^{js}} 
\end{align}
for $s \in \C$ and $g(s):=\sum_{k=1}^Kg_k^s$. 
\end{theorem}

\begin{proof}
Assume, for notational simplicity, that $K=g_1=1$. Further, for $\textnormal{Re}(s)>D_\calR(\bfs;\cdot)$, consider the sum 
\[
\zeta_\Rec(\bfs;s)-\sum_{n=0}^{d-1}s_nr^{ns}=\sum_{n=0}^{\infty}s_{n+d}r^{(n+d)s}.
\] 
We then have successively:
\begin{align*}
\zeta_\Rec(\bfs;s)-\sum_{n=0}^{d-1}s_nr^{ns}	&= a_1\sum_{n=0}^{\infty}s_{n+d-1}r^{(n+d)s}
	+\cdots +a_d\sum_{n=0}^{\infty}s_{n}r^{(n+d)s} \\
	&= a_1r^s\sum_{n=0}^{\infty}s_{n+d-1}r^{(n+d-1)s}
	+\cdots +a_dr^d\sum_{n=0}^{\infty}s_{n}r^{ns} \\
	&=a_1r^s\left(\zeta_\Rec(\bfs;s)-\sum_{n=0}^{d-2}s_nr^{ns}\right)
	+\cdots
	+a_dr^{ds}\zeta_\Rec(\bfs;s).
\end{align*}
So, 
\begin{align*}		
\zeta_\Rec(\bfs;s)\left(1-\sum_{j=}^da_jr^{js}\right)&=\sum_{n=0}^{d-1}s_nr^{ns}
-\sum_{l=0}^{d-1}a_lr^{ls}\sum_{n=0}^{d-1-l}s_nr_{ns}.
\end{align*}
Therefore, for $\textnormal{Re}(s)>D_\calR(\bfs;\cdot)$, we have
\begin{align}
\label{eqn:RecursiveStringZetaFunctionProof}
\zeta_\Rec(\bfs;s)	&=\frac{\displaystyle\sum_{n=0}^{d-1}s_nr^{ns} - \sum_{l=1}^{d-1}a_lr^{ls}\sum_{n=0}^{d-1-l}s_nr^{ns}}
{\displaystyle 1-\sum_{j=1}^da_jr^{js}}.
\end{align}
Since the right-hand side of \eqref{eqn:RecursiveStringZetaFunctionProof} defines a meromorphic function in all of $\C$, it follows that $\zeta_\Rec(\bfs;\cdot)$ is meromorphic in $\C$ and is still given by the same expression on $\C$.

Finally, for the case where $K \neq 1$ or $g_1 \neq 1$, the reasoning is exactly the same as above except for the fact that the factor $g(s)=\sum_{k=1}^Kg_k^s$ would need to be included throughout the proof, accordingly.
\end{proof}

As we have seen, the right-hand side of \eqref{eqn:RecursiveStringZetaFunction} 
allows for a meromorphic continuation of $\zeta_\calR(\bfs;\cdot)$ to $\C$. This, in turn, allows for the following two results.

\begin{corollary}
\label{cor:RecursiveStringsNearlyLattice}
Let $\calR$ be a linear recurrence relation and let $\eta_\Rec (\bfs;\cdot)$ be a recursive string as in Theorem \ref{thm:RecursiveStringZetaFunction}. Then there is a generalized lattice string $\eta$ and an entire function $h_\calR(\bfs;s)$ such that, after meromorphic continuation, for all $s \in \C$, we have 
\begin{align}
\label{eqn:RecursiveStringsNearlyLattice}
\zeta_\Rec(\bfs;s)	&= h_\calR(\bfs;s)\zeta_\eta(s).
\end{align}
Moreover, $D_\calR(\bfs)=D_\eta$ and $\calD_\calR(\bfs) \subset \calD_\eta$.
\end{corollary}

\begin{proof}
For some $d \in \N$, let $\Rec$ be a linear recurrence relation with kernel $\bfa=(a_1,\ldots,a_d)$. Given a recursive string $\eta_\Rec (\bfs;\cdot)$ determined by the gaps $\bfg=(g_1,\ldots,g_K)$, initial conditions $\bfs=(s_0,\ldots,s_{d-1})$, and multiplicative generator $r$, consider the generalized lattice string $\eta'$ determined by the initial multiplicities $\bfm'=\bfa$, initial scaling ratios $\bfr'=(r,r^2,\ldots,r^d)$ and a single gap of length 1 (i.e., $\bfg'=(1)$). We immediately have that $D_\calR(\bfs)=D_{\eta'}$.

Now, define $h_\calR(\bfs;s)$ by
\[
h_\calR(\bfs;s):= g(s) \cdot \left(\displaystyle\sum_{n=0}^{d-1}s_nr^{ns} - \sum_{l=1}^{d-1}a_lr^{ls}\sum_{n=0}^{d-1-l}s_nr^{ns}\right),
\] 
where $g(s):=\sum_{k=1}^Kg_k^s$. Then $h_\calR(\bfs;\cdot)$ is entire, and by applying Theorems \ref{thm:ClosedFormOfGeometricZetaFunction} and \ref{thm:RecursiveStringZetaFunction} to $\eta'$ and $\eta_\Rec (\bfs;\cdot)$, respectively, we see that \eqref{eqn:RecursiveStringsNearlyLattice} holds with $\eta:=\eta'$. Moreover, after meromorphic continuation, \eqref{eqn:RecursiveStringsNearlyLattice} holds for all $s \in \C$, and we conclude that $\calD_\calR(\bfs) \subset \calD_{\eta'}$.
\end{proof}

Note that in Corollary \ref{cor:RecursiveStringsNearlyLattice}, we do not conclude that $\calD_\calR(\bfs)=\calD_{\eta'}$ in general. This is due to the fact that some of the zeros of $h_\calR(\bfs;\cdot)$ might cancel some of the poles (of the meromorphic extension) of $\zeta_{\eta}=\zeta_{\eta'}$. 

The following corollary is an immediate consequence of the combination of Theorems \ref{thm:CriterionForMinkowskiMeasurability} and \ref{thm:RecursiveStringZetaFunction}.

\begin{corollary}
\label{cor:RecursiveStringMinkowski}
Let $\Omega$ be an ordinary fractal string with lengths 
\begin{align*}
%\label{eqn:LengthsRecursiveString}
\calL_\calR	&= \bigcup_{k=1}^K\{g_kr^n : r^n \text{ has multiplicity } m_n, n\in\N\cup\{0\} \} ,
\end{align*}
where $0<r<1$, $\bfg=(g_1,\ldots,g_K) \in (0,\infty)^K$, and the multiplicities $\{m_n\}$ satisfy a linear recurrence relation $\calR$ with kernel $\bfa=(a_1,\ldots,a_d)$. Furthermore, if none of the complex roots of $h_\calR(\bfs;s)=0$ is also a complex root of the Moran equation $1=\sum_{j=1}^da_jr^{js}$, then $\partial\Omega$ is not Minkowski measurable.
\end{corollary}

The discussion of recursive strings concludes with an example of a recursive string which is not a generalized self-similar string (hence, it is also not a generalized lattice string) and to which Theorem \ref{thm:RecursiveStringZetaFunction}, Corollary \ref{cor:RecursiveStringsNearlyLattice}, and Corollary \ref{cor:RecursiveStringMinkowski} apply. 

\begin{example}[The Lucas string]
\label{eg:LucasRecursiveNotLattice}
Consider the recursive string $\eta_{\textnormal{Luc}}$, called the \emph{Lucas string}, determined by the kernel $\bfa=(1,1)$, initial conditions $\bfs_{\textnormal{Luc}}:=(2,1)$, a single gap determined by $K=g_1=1$, and multiplicative generator $r=1/2$. The geometric zeta function of $\eta_{\textnormal{Luc}}$ satisfies
\begin{align*}
%\label{eqn:LucasRecursiveNotLatticeGeometricZetaFunction}
\zeta_{\textnormal{Luc}}(s)	&=\sum_{n=0}^\infty s_n2^{-ns},
\end{align*}
where $\textnormal{Re}(s)$ is large enough and the sequence $\{s_n\}_{n=0}^\infty$ satisfies the linear recurrence relation determined by $\bfa$ and with initial conditions $\bfs_{\textnormal{Luc}}$. That is, $s_0=2, s_1=1$, and for $n \geq 2$, $s_n$ satisfies the Fibonacci recursion relation \eqref{eqn:FibonacciLinearRecurrenceRelation}. The closed form of $\zeta_{\textnormal{Luc}}$ given by Theorem \ref{thm:RecursiveStringZetaFunction}, which is also a consequence of Corollary \ref{cor:RecursiveStringsNearlyLattice}), is 
\begin{align}
\label{eqn:LucasRecursiveNotLatticeZetaFunctionClosedForm}
\zeta_{\textnormal{Luc}}(s)
=\frac{2-2^{-s}}{1-2^{-s}-4^{-s}}, \quad s \in \C.
\end{align} 

Note that the sequence $\{s_n\}_{n=1}^\infty=\{1,3,4,7,\ldots\}$ (where the index $n$ begins at 1) is the sequence of \emph{Lucas numbers}. Moreover, the Fibonacci numbers $\{F_n\}_{n=0}^\infty$ and the sequence $\{s_n\}_{n=0}^\infty$ both satisfy the linear recurrence relation determined by the kernel $\bfa=(1,1)$, but with initial conditions $\bfs_{\textnormal{Fib}}:=(1,1)$ and $\bfs_{\textnormal{Luc}}=(2,1)$, respectively. (See Example \ref{eg:TheFibonacciString}.) Also, note that $h_\calR(\bfs;s)=2-2^{-s}=0$ if and only if $2^{-s}=2$, which implies $1-2^{-s}-4^{-s}=-5\neq0$. Thus, the set of complex dimensions of the Lucas string, denoted $\calD_{\textnormal{Luc}}$, coincides with $\calD_{\textnormal{Fib}}$, the set of complex dimensions of the Fibonacci string given by \eqref{eqn:FibonacciComplexDimensions} and which are the roots of the Dirichlet polynomial equation $2^{-s}+4^{-s}=1$.  Also note that $D_{\textnormal{Luc}}=D_{\textnormal{Fib}}=\log_2\varphi $, where $ \varphi =(1+\sqrt{5})/2 $ is the Golden Ratio.

Moreover, $\{s_n\}_{n=0}^\infty$ is a sequence of positive integers and $\zeta_{\textnormal{Luc}}(1)$ is positive and finite. Hence, $\zeta_{\textnormal{Luc}}$ is the geometric zeta function of a suitably defined ordinary fractal string $\Omega_{\textnormal{Luc}}$. However, such a fractal string is \emph{not} lattice, in the sense that it cannot be realized as the complement of the attractor of a (lattice or even nonlattice) self-similar system in some closed interval (see Definition \ref{def:SelfSimilarLatticeOrdinaryFractalStrings}). Indeed, $\zeta_{\textnormal{Luc}}$ can not be put into the form given in \eqref{eqn:GeometricZetaFunctionOrdinarySelfSimilarString} as indicated by the fact that, in \eqref{eqn:LucasRecursiveNotLatticeZetaFunctionClosedForm}, the numerator $h_\calR(\bfs;s)=2-2^{-s}$ can not be written in the form $L^s\sum_{k=1}^Kg_k^s$ where $L$ and each $g_k$ are positive real numbers. Hence, Theorem \ref{thm:LatticeNonlatticeDichotomy} does not apply. However, Corollary \ref{cor:RecursiveStringMinkowski} applies and, hence, $\partial\Omega_{\textnormal{Luc}}$ is not Minkowksi measurable.
\end{example}

\begin{example}[Recursive structure of generalized Cantor strings]
\label{eg:RecursiveStructureOfGeneralizedCantorStrings}
Every generalized Cantor string $\mu$ is a recursive string; see Example \ref{eg:GeneralizedCantorStrings} and \cite[Chs.~8 \& 10]{LapvF6}. Indeed, we have
\begin{align*}
\mu		&=\sum_{n=0}^\infty s_n\delta_{\lbrace r^{-n}\rbrace},
\end{align*}
where the sequence $\{s_n\}_{n=0}^\infty$ is given by $s_n=b^n$ for every $n$ and satisfies the linear recurrence relation $\Rec$ determined by the kernel $\bfa=(b)$ with initial condition $\bfs=(1)$.  By Theorem \ref{thm:RecursiveStringZetaFunction}, the meromorphic extension of the geometric zeta function $\zeta_\mu$ given in \eqref{eqn:GeometricZetaFunctionOfGeneralizedCantorString} is recovered as $\zeta_\Rec(\bfs;\cdot)$. That is, 
\begin{align*}
\zeta_\mu(s)	&= \zeta_\calR(\bfs;s) =\frac{1}{1-b\cdot r^{-s}}, \quad s \in \C.
\end{align*}
\end{example}

The technique developed in the next section is motivated by the theory of complex dimensions of fractal strings and is designed to perform a multifractal analysis of self-similar measures.

\section{Multifractal analysis via scaling regularity and scaling zeta functions}
\label{sec:MultifractalAnalysisViaScalingRegularityAndScalingZetaFunctions}

In this section, an approach to multifractal analysis is described in which, for a given weighted self-similar system $(\bfPhi,\bfp)$, a family of fractal strings is defined by gathering lengths (or rather, scales) according to scaling regularity values $\alpha$. Then, a family of \emph{$\alpha$-scaling zeta functions} is readily defined and the collection of their abscissae of convergence provides a multifractal spectrum of dimensions, called the \emph{scaling multifractal spectrum}, which is akin to the geometric and symbolic multifractal spectra (see Section \ref{sec:MultifractalSpectra}). 

The approach given in this section generalizes the results regarding \emph{partition zeta functions} found in \cite[\S5]{ELMR}. Specifically, the measures studied in \cite[\S5]{ELMR} are limited to measures supported on subsets of $[0,1]$, whereas the framework provided below allows for results on self-similar measures which are supported on self-similar sets in Euclidean spaces of any dimension. See Section \ref{sec:MultifractalAnalysisOfSelfSimilarSystems} as well as \cite{Hut}. 

We refer to Section 13.3 (of the second edition) of \cite{LapvF6} for a survey of aspects of the theory of multifractal zeta functions (and related partition zeta functions), as developed in \cite{LLVR} and \cite{LapRo1}. We also refer to \cite{ELMR} and \cite{Rock} where additional results can be found.

\subsection{$\alpha$-scales and $\alpha$-scaling zeta functions}
\label{sec:AlphaScalesAndAlphaScalingZetaFunctions}

Throughout this section, and indeed throughout the paper, only weighted self-similar systems $(\bfPhi,\bfp)$ with scaling ratios $\bfr=\{r_j\}_{j=1}^N$ satisfying the open set condition are considered; see Section \ref{sec:MultifractalAnalysisOfSelfSimilarSystems}. 

\begin{definition}
\label{def:SequenceOfAlphaScales}
Let $(\bfPhi,\bfp)$ be a weighted self-similar system. For a scaling regularity value $\alpha \in \R$, the \emph{sequence of $\alpha$-scales}, denoted  $\calL_{\bfr,\bfp}(\alpha)$, is the fractal string given by
\begin{align*}
%\label{eqn:AlphaScalesDefinition}
\calL_{\bfr,\bfp}(\alpha)	&=\left\{ r_J : J \in \scrJ, A_{\bfr,\bfp}(J) =\alpha \right\},
\end{align*}
where $A_{\bfr,\bfp}$ is defined in Definition \ref{def:ScalingRegularity}.
\end{definition}

Alternately, the \emph{distinct $\alpha$-scales} are denoted by $l_n(\alpha)$ and the corresponding \emph{multiplicities} are denoted by $m_n(\alpha)$. Thus,
\begin{align*}
m_n(\alpha)	&:=\#\{J \in \scrJ: r_J=l_n(\alpha),A_{\bfr,\bfp}(J) =\alpha\},
\end{align*}
and we can consider the sequence of $\alpha$-scales to be given by 
\begin{align*}
%\label{eqn:DistinctAlphaScalesWithMultiplicity}
\calL_{\bfr,\bfp}(\alpha)	&=\left\{ l_n(\alpha) : l_n(\alpha) \textnormal{ has multiplicity } m_n(\alpha) \right\}.
\end{align*}

Given that $\calL_{\bfr,\bfp}(\alpha)$ is a fractal string when there are infinitely many words $J \in \scrJ$ such that $A_{\bfr,\bfp}(J)=\alpha$ for some $\alpha \in \R$, one defines a zeta function and a set of complex dimensions for each such $\alpha$ as follows.

\begin{definition}
\label{def:ScalingZetaFunctionAndScalingMultifractalSpectrum}
Consider $\alpha \in \R$ such that $\calL_{\bfr,\bfp}(\alpha)$ is not empty. The \emph{$\alpha$-scaling zeta function} $\zeta_{\bfr,\bfp}(\alpha;\cdot)$ is the geometric zeta function of the sequence of $\alpha$-scales $\calL_{\bfr,\bfp}(\alpha)$. That is,
\begin{align}
\label{eqn:AlphaScalingZetaFunction}
\zeta_{\bfr,\bfp}(\alpha;s) &:= \zeta_{\calL_{\bfr,\bfp}(\alpha)}(s) =
\sum_{A_{\bfr,\bfp}(J)=\alpha}r_J^s,
\end{align}
where $\textnormal{Re}(s)>D_{\bfr,\bfp}(\alpha):=D_{\calL_{\bfr,\bfp}(\alpha)}$. If $\calL_{\bfr,\bfp}(\alpha)$ is empty, we set $\zeta_{\bfr,\bfp}(\alpha;s)=0$. The \emph{scaling mulitfractal spectrum} $f_{\bfr,\bfp}(\alpha)$ is the function given by the maximum of $0$ and the abscissa of convergence of $\zeta_{\bfr,\bfp}(\alpha;s)$ for each $\alpha\in\R$. That is,
\begin{align*}
%\label{eqn:AbscissaOfConvergenceFunction}
f_{\bfr,\bfp}(\alpha)	&:= \max \left\{0, \inf \left\{ \sigma \in \R : \zeta_{\bfr,\bfp}(\alpha;\sigma) < \infty \right\}\right\},
\end{align*}
for $\alpha \in \R$. 
\end{definition}

\noindent More precisely, $f_{\bfr,\bfp}(\alpha)$ is the maximum of 0 and
the abscissa of convergence of the Dirichlet series which defines the $\alpha$-scaling zeta function $\zeta_{\bfr,\bfp}(\alpha;\cdot)$. Thus, $f_{\bfr,\bfp}(\alpha) \geq 0$ whenever $\calL_{\bfr,\bfp}(\alpha)$ comprises an infinite number of scales and $f_{\bfr,\bfp}(\alpha)=0$ otherwise. Hence, $\alpha$ is \emph{nontrivial} if $\calL_{\bfr,\bfp}(\alpha)$ comprises an infinite number of scales, otherwise $\alpha$ is \emph{trivial}  (see Remark \ref{rmk:TrivialScalingRegularity} and compare Remark 4.5 of \cite{ELMR}). Accordingly, for a nontrivial scaling regularity value $\alpha$, the set given by
\begin{align*}
%\label{eqn:NontrivialScalingRegularityAbscissa}
&	\{s \in \C : \textnormal{Re}(s) > f_{\bfr,\bfp}(\alpha)\}
\end{align*}
is the largest open right half-plane on which
the Dirichlet series in \eqref{eqn:AlphaScalingZetaFunction} is absolutely convergent.

\begin{remark}
\label{rmk:TrivialScalingRegularity}
Unlike in Remark 4.5 of \cite{ELMR}, all trivial scaling regularity values $\alpha_0$ in the current setting generate an empty sequence of $\alpha_0$-scales. Thus, for any $\alpha \in \R$, $\calL_{\bfr,\bfp}(\alpha)$ is either countably infinite or empty. See Remark \ref{rmk:MultiplicitiesAndInfiniteAlphaScales} for clarification.
\end{remark}

\begin{definition}
\label{def:ScalingRegularityComplexDimensionsAndTapestry}
Let $W_\alpha \subset \C$ be a window on a connected open neighborhood of which $\zeta_{\bfr,\bfp}(\alpha;\cdot)$ has a meromorphic extension. (Again, both the geometric zeta function of $\calL_{\bfr,\bfp}(\alpha)$ and its meromorphic extension are denoted by $\zeta_{\bfr,\bfp}(\alpha,\cdot)$.) The set of (\emph{visible}) $\alpha$-\emph{scaling complex dimensions}, denoted $\calD_{\bfr,\bfp}(\alpha)$, is the set of (visible) complex dimensions of the sequence of $\alpha$-scales $\calL_{\bfr,\bfp}(\alpha)$ given by
\begin{align*}
%\label{eqn:AlphaDimension}
\calD_{\bfr,\bfp}(\alpha)	&= \left\{ \omega \in W_\alpha : \zeta_\calL \textnormal{ has a pole at } \omega \right\}.
\end{align*}
The \emph{tapestry of complex dimensions} $\calT_{\bfr,\bfp}$ with respect to the regions $W_\alpha$ is the subset of $\R \times \C$ given by
\begin{align*}
%\label{eqn:Tapestry}
\calT_{\bfr,\bfp}	&:= \left\{(\alpha,\omega): \alpha \textnormal{ is nontrivial, } \omega \in \calD_{\bfr,\bfp}(\alpha) \right\}.
\end{align*}
\end{definition}

\subsection{Scaling regularity values attained by self-similar measures}
\label{sec:ScalingRegularityValuesAttainedbySelfSimilarMeasures}

The collection of all scaling regularity values $A_{\bfr,\bfp}(J)$ attained by the words $J \in \scrJ$ with respect to a given weighted self-similar system $(\bfPhi,\bfp)$ are found readily. The following notation is used in order to facilitate the statement of the results. Note that, as throughout the paper, only weighted self-similar systems which satisfy the open set condition are considered.

\begin{notation}
\label{not:CombinationsOfVectors}
Let $N \in \N$ with $N\geq 2$. For a pair of $N$-vectors $\bfk$ and $\bfr$ with $\bfr \in (0,1)^N$ and $\bfk \in (\N\cup\{0\})^N$, let $\bfr^{\bfk}:=r_{1}^{k_1} \cdots r_{N}^{k_{N}}$. Furthermore, denote by $\gcd(\bfk)$ the greatest common divisor of the nonzero components of $\bfk$. Let $\sum\bfk:=\sum_{j=1}^Nk_j$ and
\begin{align}
\label{eqn:MultinomialCoefficients}
\binom{\sum\bfk}{\bfk}	&:= \binom{K}{k_1 \ldots k_N} = \frac{K!}{k_1! \cdots k_N!},
\end{align}
where $K=\sum\bfk$.
\end{notation}

Recall from Notation \ref{not:SelfSimilarityWords} that $\mathscr{J}$ denotes the collections of all finite words on the alphabet $\{1,\ldots,N\}$.

\begin{lemma}
\label{lem:PreliminaryAttainedScalingRegularityAndMultiplicities}
Let $(\bfPhi,\bfp)$ be a weighted self-similar system as above and let $J \in \scrJ$ be a nonempty word. Then there exist a unique vector $\bfk \in (\N\cup\{0\})^N$ with $\gcd(\bfk)=1$ and a unique positive integer $n$ such that $|J|=n\sum\bfk$ and 
\begin{align}
\label{eqn:DeterminingScalingRegularity}
A_{\bfr,\bfp}(J)	&= \alpha(\bfk) := \log_{\bfr^\bfk}\bfp^\bfk.
\end{align}
Moreover, for each $\bfk \in (\N\cup\{0\})^N$ we have
\begin{align*}
%\label{eqn:NumberOfWordsJPerVectork}
\#\left\{J \in \scrJ_{\sum\bfk} : \#\left\{q: \pi_q(J)=j \right\} = k_j, j\in \left\{1,\ldots,N \right\} \right\}		&= \binom{\sum\bfk}{\bfk},
\end{align*}
where $\binom{\sum\bfk}{\bfk}$ is given by \eqref{eqn:MultinomialCoefficients}.
\end{lemma}

\begin{proof}
The results follow immediately from basic combinatorics and the definitions of $r_J$ and $p_J$ given in Notation \ref{not:ProductsOfScalingRatios}.
\end{proof}

\begin{remark}
\label{rmk:MultiplicitiesAndInfiniteAlphaScales} 
An immediate consequence of Lemma \ref{lem:PreliminaryAttainedScalingRegularityAndMultiplicities} is that each $\bfk \in (\N\cup\{0\})^N$ satisfying $\gcd(\bfk)=1$ yields a countably infinite collection of words $J \in \scrJ$ which have scaling regularity $A_{\bfr,\bfp}(J)=\alpha(\bfk)$. The scales $r_J$ associated with these words essentially constitute the terms of the Dirichlet series which defines the $\alpha$-scaling zeta function (for $\alpha=\alpha(\bfk)$, as in \eqref{eqn:DeterminingScalingRegularity}) in Definition \ref{def:ScalingZetaFunctionAndScalingMultifractalSpectrum}. 

Also, to clarify Remark \ref{rmk:TrivialScalingRegularity}, suppose $\alpha_0 \in \R$. If there is a word $J \in \scrJ$ such that $A_{\bfr,\bfp}(J)=\alpha_0$, then by Lemma \ref{lem:PreliminaryAttainedScalingRegularityAndMultiplicities}, $\alpha_0$ is nontrivial. Indeed, in that case, for each $n \in \N$, there is a word $H_n \in \scrJ$  such that $|H_n|=n|J|$ and $A_{\bfr,\bfp}(H_n)=A_{\bfr,\bfp}(J)=\alpha_0$.

Furthermore, for each $n \in \N$ we have 
\begin{align*}
%\label{eqn:NumberOfWordsJPerVectornk}
\#\left\{J \in \scrJ_{n\sum\bfk} : \#\left\{q: \pi_q(J)=j \right\} = nk_j, j\in \left\{1,\ldots,N \right\} \right\}		&= \binom{n\sum\bfk}{n\bfk}.
\end{align*}
If there are $\bfm_1,\bfm_2 \in (\N\cup\{0\})^N$ such that $\gcd(\bfm_1)=\gcd(\bfm_2)=1$, $\bfm_1 \neq \bfm_2$, and $\alpha(\bfm_1)=\alpha(\bfm_2)$, then the total multiplicity of a given scale $r'$ is given by the sum of all multinomial coefficients $\binom{n\sum\bfk}{n\bfk}$   where $r'=\bfr^{n\bfk}$. The determination of a closed formula for such total multiplicity is, in general, a difficult problem. However, various special cases are addressed in the remainder of this section and in Section \ref{sec:FurtherResultsAndFutureWork}. In particular, see Theorems \ref{thm:NearlyRecoverScalingZetaFunction} and \ref{thm:FullFamilyOfScalingZetaFunctions}, along with Corollary \ref{cor:RecoverBesicovitchSubsetResults}.
\end{remark}

\subsection{Self-similar and lattice structures within self-similar measures}
\label{sec:SelfSimilarAndLatticeStructuresWithinSelfSimilarMeasures}

For a weighted self-similar system $(\bfPhi,\bfp)$ as above, let $\alpha_j:=\log_{r_j}p_j$ for each $j=1,\ldots,N$ denote the scaling regularity values attained with respect to $(\bfPhi,\bfp)$ by all words $J$ such that $|J|=1$. Let $N_0$ denote the number of distinct values among the $\alpha_j$. Then $N_0 \leq N$ and we denote these distinct $\alpha_j$ values by $\beta_q$, where $q=1,\ldots,N_0$. 

\begin{theorem}
\label{thm:NearlyRecoverScalingZetaFunction}
Let $(\bfPhi,\bfp)$ be a weighted self-similar system as above. Suppose the collection of distinct scaling ratios $\{\beta_q: q=1,\ldots,N_0\}$ is rationally independent. Then, for each $q \in \{1,\ldots,N_0\}$, $\zeta_{\bfr,\bfp}(\beta_q;\cdot)$ has a meromorphic continuation to all of $\C$ given by
\begin{align}
\label{eqn:NearlyRecoverScalingZetaFunction}
\zeta_{\bfr,\bfp}(\beta_q;s)=\frac{\displaystyle \sum_{j:\alpha_j=\beta_q}r_j^s}{\displaystyle 1-\sum_{j:\alpha_j=\beta_q}r_j^s}, \quad \textnormal{for } s \in \C.
\end{align}
\end{theorem}

\begin{proof}
Since the $\beta_q$ values are rationally independent, the only words $J$ with regularity $\beta_q$ are those with components $\pi_n(J)=j$ where $\alpha_j=\beta_q$. Thus, for each $q \in \{1,\ldots,N_0\}$, the $\beta_q$-scaling zeta function is equal to the scaling zeta function of a self-similar system $\bfPhi_q$ whose scaling ratios are given by $\{r_j:\alpha_j=\beta_q\}$ less the first term which corresponds to the empty word. So, by Theorem \ref{thm:ClosedFormOfGeometricZetaFunction} we have 
\begin{align*}
%\label{eqn:ProofNearlyRecoverScalingZetaFunction}
\zeta_{\bfr,\bfp}(\beta_q;s)	&= \zeta_{\bfPhi_q}(s)-1 = \frac{\displaystyle \sum_{j:\alpha_j=\beta_q}r_j^s}{\displaystyle 1-\sum_{j:\alpha_j=\beta_q}r_j^s}.
\end{align*}
\end{proof}

\begin{corollary}
\label{cor:NearlyRecoverGeneralizedLatticeString}
If the conditions of Theorem \ref{thm:NearlyRecoverScalingZetaFunction} are satisfied and, additionally, if there exists a unique $r$ such that $0<r<1$ and for each $j$ where $\alpha_j=\beta_q$ we have $r_j=r^{u_j}$ for some $u_j \in \N$, then there is a generalized lattice string $\eta$ such that
\begin{align*}
%\label{eqn:LatticeInBetaq}
\zeta_{\bfr,\bfp}(\beta_q;s) = h(s)\zeta_\eta(s),
\end{align*}
where $\displaystyle h(s):=\sum_{j:\alpha_j=\beta_q}r^{u_js}$ and $\eta:=\eta_\bfPhi$ \emph{(}as in Definition \ref{def:ScalingMeasure}\emph{)}.
\end{corollary}

\begin{remark}
\label{rmk:ComplexDimensionsWithSCalingRegularity}
For each $q \in \{1,\ldots,N_0\}$, the function $h(s):=\sum_{j:\alpha_j=\beta_q}r_j^s$ in the numerator of the right-hand side of \eqref{eqn:NearlyRecoverScalingZetaFunction} is entire. Moreover, the complex roots of $h(s)=0$ are distinct from the poles of $\zeta_{\bfPhi_q}$, so $\calD(\beta_q)=\calD_{\bfPhi_q}$. 
Also, note that the hypotheses of Corollary \ref{cor:NearlyRecoverGeneralizedLatticeString} do not require the self-similar system $\bfPhi$ to be lattice. Section \ref{sec:ASelfSimilarSystemWithNonlatticeAndLatticeStructure} examines a self-similar measure built upon a nonlattice self-similar system which satisfies the hypotheses of Corollary \ref{cor:NearlyRecoverGeneralizedLatticeString} for a particular scaling regularity value.
\end{remark}

The $\alpha$-scaling zeta functions with respect to regularity values $\alpha$ which do not fit any of the conditions required in the results presented in this section are much harder to determine in general; see Section \ref{sec:ASelfSimilarSystemWithNonlatticeAndLatticeStructure}. However, given particular constraints, the full family of $\alpha$-scaling zeta functions for certain weighted self-similar systems can be determined, as we shall see in the next section.

\subsection{Full families of $\alpha$-scaling zeta functions}
\label{sec:FullFamiliesOfAlphaScalingZetaFunctions}

The development that follows determines the $\alpha$-scaling zeta functions associated with a weighted self-similar system $(\bfPhi,\bfp)$ where the distinct regularity values $\beta_q$ are rationally independent and for each $q$, the corresponding scaling ratios $\{r_j:\alpha_j=\beta_q\}$ are given by a single value $t_q$ such that $0<t_q<1$. In this setting, but unlike in Section \ref{sec:SelfSimilarAndLatticeStructuresWithinSelfSimilarMeasures}, the full family of $\alpha$-scaling zeta functions as well as all of the corresponding abscissae of convergence are determined.

Let $c_q=\#\{j:\alpha_j=\beta_q\}$ and $\bfc:=(c_1,\ldots,c_{N_0})$. Then $\bfc \in \N^{N_0}$ and $\displaystyle \sum\bfc=N$. For any pair of $N_0$-vectors $\bfm$ and $\bft$, let $\bft^{\bfm}:=t_{1}^{m_1} \cdots t_{N_0}^{m_{N_0}}$. The multinomial theorem implies that for each $n,K \in \N$ we have 
\begin{align*}
%\label{eqn:MultinomialTheoremImplication}
N^n	&= \sum_{\sum\bfk=K}\binom{nK}{n\bfk}=\sum_{\sum\bfm=K}\binom{nK}{n\bfm}\bfc^{n\bfm}.
\end{align*}
As shown in the following theorem, the products of the form $\binom{nM}{n\bfm_0}\bfc^{n\bfm_0}$, where $n \in \N$, $M=\sum\bfm_0$, and $\gcd(\bfm_0)=1$, are the multiplicities of the $\alpha$-scales of the scaling zeta functions associated with $(\bfPhi,\bfp)$ (cf.~\cite[\S 5]{ELMR}, especially Theorem 5.2, Lemma 5.10, Proposition 5.11, and Theorem 5.12 therein). Also, the following theorem determines the scaling zeta functions $\zeta_{\bfr,\bfp}$ and scaling multifractal spectrum $f_{\bfr,\bfp}$ associated with a weighted self-similar system $(\bfPhi,\bfp)$.

\begin{theorem}
\label{thm:FullFamilyOfScalingZetaFunctions}
Let $(\bfPhi,\bfp)$ be a weighted self-similar system. For each $q=1,\ldots,N_0$, suppose there exists a unique $t_q$ such that $0<t_q<1$ and for each $j=1,\ldots,N$ such that $\alpha_j=\beta_q, r_j=t_q$. Further, suppose the collection $\{\beta_q\}_{q=1}^{N_0}$ is rationally independent. Then there exists a unique vector $\bfv=(v_1,\ldots,v_{N_0})$ such that $\beta_q=\log_{t_q}v_q$, where $p_j=v_q$ for each $j=1,\ldots,N$ such that $\alpha_j=\beta_q, r_j=t_q$ and some $q \in 1,\ldots,N_0$. 

Furthermore, the distinct scaling regularity values attained with respect to $(\bfPhi,\bfp)$ are given by 
\begin{align*}
%\label{eqn:DistinctRegularity}
\beta(\bfm)	&:=\log_{\bft^{\bfm}}\bfv^{\bfm},
\end{align*}
for some $\bfm \in (\N\cup\{0\})^{N_0}$ where $\gcd(\bfm)=1$ and $\bfm \neq \mathbf{0}$. Also, for $\bfm \in (\N\cup\{0\})^{N_0}$ where $\gcd(\bfm)=1$ and $\bfm \neq \mathbf{0}$ and for each $n \in \N$, the number of ways a scaling regularity value $\beta(n\bfm)=\beta(\bfm)$ is attained with respect to  $(\bfPhi,\bfp)$ at level $nM$, where $M:=\sum\bfm=\sum_{q=1}^{N_0}m_q$, is given by
\begin{align*}
%\label{eqn:MultiplicityOfAGivenRegularity}
&	\binom{nM}{n\bfm}\bfc^{n\bfm}.
\end{align*}

Moreover, if $\bfm \in (\N\cup\{0\})^{N_0}$ where $\gcd(\bfm)=1$ and $\bfm \neq \mathbf{0}$, then 
\begin{align*}
%\label{eqn:FullFamilyOfScalingZetaFunctions}
\zeta_{\bfr,\bfp}(\beta(\bfm);s)	&= \sum_{n=1}^{\infty}\binom{nM}{n\bfm}\bfc^{n\bfm}\bft^{sn\bfm}
\end{align*}
where $M:=\sum\bfm$, $\textnormal{Re}(s)>f_{\bfr,\bfp}(\beta(\bfm))$, and the corresponding abscissa of convergence $f_{\bfr,\bfp}(\beta(\bfm))$ of $\zeta_{\bfr,\bfp}(\beta(\bfm);\cdot)$ is given by
\begin{align*}
%\label{eqn:FullFamilyOfAbscissaeOfConvergence}
f_{\bfr,\bfp}(\beta(\bfm))	&= \log_{\bft^{\bfm}}\left(\frac{\bfm^{\bfm}}{M^M\bfc^{\bfm}}\right).
\end{align*}
\end{theorem}

\begin{proof}
Fix $q \in \{1,\ldots,N_0\}$. For each $j=1,\ldots,N$ such that $\alpha_j=\beta_q$ we have $p_j=t_q^{\beta_q}$. Thus, $\bfv:=(t_1^{\beta_1},\ldots,t_{N_0}^{\beta_{N_0}})$ is the desired vector. 

Let $J \in \scrJ$. By Lemma \ref{lem:PreliminaryAttainedScalingRegularityAndMultiplicities}, there exist $n \in \N$ and $\bfk \in (\N \cup \{0\})^N$ such that $\gcd(\bfk)=1$, $|J|=n\sum\bfk$, and $A_{\bfPhi,\bfp}=\log_{\bfr^\bfk}\bfp^\bfk$. 

For each $q \in \{1,\ldots,N_0\}$, let $m_q:=\sum_{j:\alpha_j=\beta_q}k_j$. Define $\bfm:=(m_1,\ldots,m_{N_0})$.  Then $\sum\bfm=\sum\bfk$, $\bft^\bfm=\bfr^\bfk$, and $\bfv^\bfm=\bfp^\bfk$. Hence $A_{\bfPhi,\bfp}(J)=\beta(\bfm):=\log_{\bft^{\bfm}}\bfv^{\bfm}$. Now, if $\gcd(\bfm_0)=1$, then $\beta(\bfm)=\beta(\bfm_0)$ if and only if there is $n \in \N$ such that $\bfm=n\bfm_0$ since (by hypothesis) the collection $\{\beta_q\}_{q=1}^{N_0}$ is rationally independent. Moreover, for $M:=\sum\bfm_0$ we have
\begin{align*}
%\label{eqn:CountingBetam0Multiplicities}
\#\{J\in \scrJ_{n\sum\bfm_0}:A_{\bfPhi,\bfp}=\beta_{\bfm_0}\}	&=\binom{nM}{n\bfm_0}\bfc^{n\bfm_0}.
\end{align*}

Note that each $J \in \scrJ$ will be taken into account since all $\bfm \in (\N \cup \{0\})^{N_0}$ where $\bfm \neq 0$ and $\gcd(\bfm)=1$ are considered. 

Now, if $A_{\bfPhi,\bfp}(J)=\beta(\bfm)$, then $r_J=\bft^{n\bfm}$ for some $n \in \N$. Furthermore, for each $n \in \N$
\begin{align*}
\#\{J\in \scrJ_{n\sum\bfm_0}:A_{\bfPhi,\bfp}=\beta_{\bfm_0}\}	&=\binom{nM}{n\bfm_0}\bfc^{n\bfm_0}.
\end{align*}
Hence,
\begin{align*}
\zeta_{\bfr,\bfp}(\beta(\bfm);s)	&= \sum_{n=1}^{\infty}\binom{nM}{n\bfm}\bfc^{n\bfm}\bft^{sn\bfm}.
\end{align*}

In order to determine the abscissa of convergence $f_{\bfr,\bfp}(\beta(\bfm))$, apply Stirling's formula and make use of the function $\psi(s)$ for $s \in \R$ given by
\begin{align*}
%\label{eqn:FunctionForStirlingsFormula}
\psi(s)	&= \frac{\bft^{s\bfm}c^\bfm M^M}{\bfm^\bfm}.
\end{align*}
The derivative of $\psi$ satisfies $\psi'(s)<0$ since $0<t_q<1$ for $q=1,\ldots,N_0$. Also, $\psi(0)=\bfc^\bfm M^M/\bfm^\bfm>1$ since $\bfm \neq \mathbf{0}$, $M=\sum\bfm$, and $c_q>0$ for each $q=1,\ldots,N_0$. Hence, there is a real unique real number $\rho$ such that $\rho>0$ and 
\begin{align*}
%\label{eqn:UniqueAbscissaOfConvergence}
1&=\frac{\bft^{\rho\bfm}\bfc^\bfm M^M}{\bfm^\bfm}.
\end{align*}
The real number $\rho$ will prove to be our abscissa of convergence. Indeed, for a fixed real number $s$, Stirling's formula yields 
\begin{align*}
	\binom{nM}{n\bfm}\bfc^{n\bfm}\bft^{sn\bfm} 
	&=\frac{(nM)!}{(nm_1)!\cdots (n_m{N_0})}\bfc^{n\bfm}\bft^{sn\bfm}\\
	&=\frac{\bfc^{n\bfm}\bft^{sn\bfm}M^{nM}}{\bfm^{n\bfm}} \cdot \frac{\sqrt{2\pi nM}}{\bfm^{n\bfm}}(1+\varepsilon_n) , 
\end{align*}
where $\varepsilon_n \rightarrow 0$ as $n \rightarrow \infty$. Hence, 
\[
\left[\binom{nM}{n\bfm}\bfc^{n\bfm}\bft^{sn\bfm}\right]^{1/n} =\frac{\bfc^{\bfm}\bft^{s\bfm}M^{M}}{\bfm^{\bfm}}(1+\delta_n), 
\]
where $\delta_n \rightarrow 0$ as $n \rightarrow \infty$. The root test implies that the numerical series 
\begin{align*}
\zeta_{\bfr,\bfp}(\beta(\bfm);s)	&= \sum_{n=1}^{\infty}\binom{nM}{n\bfm}\bfc^{n\bfm}\bft^{sn\bfm}
\end{align*}
converges for $s>\rho$ and diverges for $s<\rho$. Therefore,
\[
f_{\bfr,\bfp}(\beta(\bfm))=\rho=\log_{\bft^\bfm}\left(\frac{\bfm^\bfm}{\bfc^\bfm M^M}\right).
\]
\end{proof}

\begin{corollary}
\label{cor:RecoverBesicovitchSubsetResults}
If the conditions of Theorem \ref{thm:FullFamilyOfScalingZetaFunctions} are satisfied and, additionally, if $N_0=N$, then Proposition \ref{prop:HausdorffDimensionsOfBesicovitchSubsets} is recovered for scaling regularity values $\alpha(\bfk)$. Specifically, in this case we have
\begin{align*}
%\label{eqn:RecoverHausdorffDimensionsOfBesicovitchSubsets}
f_{\bfr,\bfp}(\alpha(\bfk))	&= \log_{\bfr^{\bfk}}\left(\frac{\bfk^{\bfk}}{K^K}\right) = \frac{\sum_{j=1}^N (k_j/k) \log{k_j/K}}{\sum_{j=1}^N (k_j/K) \log{r_j}} =  \dim_H(F(\bfk/K)),
\end{align*}
where $F(\bfk/K) $ is the Besicovitch subset \emph{(}as given in Definition \ref{def:BesicovitchSubsets}\emph{)} of the self-similar set $F$.
\end{corollary} 

\begin{remark}
\label{rmk:RecoverPZFMSTCDResults}
In addition to the recovery of Proposition \ref{prop:HausdorffDimensionsOfBesicovitchSubsets}, Theorem \ref{thm:FullFamilyOfScalingZetaFunctions} recovers and allows for the generalization of all of the results found in Section 5 of \cite{ELMR} where the family of partitions is taken to be the \emph{natural} family of partitions (adapted to $\bfPhi$). This generalization includes a partial recovery of the behavior of the multifractal spectrum discussed in Remark \ref{rmk:StructureOfMultifractalSpectrum} when $N_0=2$ and $t_1=t_2$, as in Example \ref{eg:MeasuresOnTheCantorSet}. (The recovery is partial due to the restriction to a countable collection of scaling regularity values; however, taking the concave envelope of the resulting graph yields the entire spectrum.) Moreover, due to the use of scaling regularity in place of coarse H\"{o}lder regularity, \emph{our results extend the results of Section 5 of \emph{\cite{ELMR}} to the setting of arbitrary self-similar measures supported on self-similar subsets of some Euclidean space of any dimension, as opposed to self-similar measures supported on subsets of the unit interval $[0,1]$.}
\end{remark}

\begin{remark}[Natural Hausdorff measures]
\label{rmk:NaturalHausdorffMeasures}
Consider a weighted self-similar system $(\bfPhi,\bfp)$, where $\bfp=(r_1^D,\ldots,r_N^D)$ such that $D$ is the unique real (and hence, positive) root of the corresponding Moran equation \eqref{eqn:Moran}. The resulting self-similar measure $\mu$ is the natural Hausdorff measure of the underlying fractal support. Such measures are used to find lower bounds on the Hausdorff dimension of their supports; see \cite[Ch.~9]{Falc}. Moreover, $\alpha=D$ is the only nontrivial scaling regularity value associated with $(\bfPhi,\bfp)$. One then says that $\mu$ is \emph{monofractal}. On the other hand, if $\bfp\neq (r_1^D,\ldots,r_N^D)$, then there are countably many nontrivial regularity values associated with $(\bfPhi,\bfp)$. It follows that $\mu$ is truly \emph{multifractal} in this case.

In the case of the ternary Cantor set with $\bfr=(1/3,1/3)$ weighted by $\bfp=(1/2,1/2)$, the graph of the primitive of the natural Hausdorff measure $\mu$ (i.e., the graph of $\mu([0,x]$)  is the well-known Devil's staircase; see,\cite[Plate 83, p.~83]{BM2}, \cite[Ch.~6]{Fed},\cite[\S12.1.1]{LapvF6} and Example \ref{eg:MeasuresOnTheCantorSet} above. The primitive of $\mu$ is the Cantor--Lebesgue function. Recall that this function is a nondecreasing, surjective, and continuous function from $[0,1]$ to itself which has derivative zero on $[0,1]\setminus F$, where $F$ is the Cantor set.
\end{remark}

In the next section we close the paper by proposing some preliminary ideas for further results and future work.

\section{Further results and future work}
\label{sec:FurtherResultsAndFutureWork}

The results presented in this paper suggest several interesting problems to pursue. For instance, since our results stem directly from scaling ratios and probability vectors, one may consider interpreting analogous results on an arbitrary metric space along with an appropriate space of measures. Also, the following conjecture, which was originally stated (in a similar but more restrictive setting) in Conjecture 5.8 of \cite{ELMR}, is not addressed in the rest of the present paper.

\begin{conjecture}
\label{conj:RecoverMultifractalSpectra}
For a self-similar measure $\mu$ and all $t \in [t_{\min},t_{\max}]$, we have
\begin{align*}
%\label{eqn:conj}
\hat{f}(t) = f_g(t) = f_s(t),
\end{align*}
where $\hat{f}(t)$ is the concave envelope of the scaling multifractal spectrum $f(\alpha)$ on $[t_{\min},t_{\max}]$, and $f_g$ and $f_s$ are the geometric and symbolic Hausdorff multifractal spectra defined in Section \emph{\ref{sec:MultifractalAnalysisOfSelfSimilarSystems}}.
\end{conjecture}

Another problem worthy of study, but not addressed in this paper and yet motivated by the theory of complex of dimensions of fractal strings in \cite{LapvF6}, is the determination of the full collection of the sets of $\alpha$-scaling complex dimensions $\calD_{\bfr,\bfp}(\alpha)$ of a weighted self-similar system $(\bfPhi,\bfp)$ with respect to the nontrivial regularity values $\alpha$ and, in turn, the determination of the tapestry of complex dimensions $\calT_{\bfr,\bfp}$. Aside from weighted self-similar systems with regularity values satisfying the conditions of Theorem \ref{thm:NearlyRecoverScalingZetaFunction} and Corollary \ref{cor:NearlyRecoverGeneralizedLatticeString}, the $\alpha$-scaling complex dimensions associated with other regularity values, as in Theorem \ref{thm:FullFamilyOfScalingZetaFunctions}, are not known.\footnote{However, in a related setting, partial results in some special cases motivated by \cite{LapRo1} (and now also \cite{ELMR}) are under development in the work in progress \cite{EssLap}.} 

Nonetheless, the final two sections provide interesting partial results on such further problems as well as motivation for further research.

\subsection{Generalized hypergeometric series}
\label{sec:GeneralizedHypergeometricSeries}

In this section, the scaling zeta functions found in Corollary \ref{cor:RecoverBesicovitchSubsetResults} (of Theorem \ref{thm:FullFamilyOfScalingZetaFunctions}) are shown to be generalized hypergeometric series (see \cite{BatErd}). Such series have been well studied and well understood, and related work may provide an alternative or supplementary approach to the theory of complex dimensions of self-similar measures.

Let $b>0$, $\bfk = (k_1,\ldots,k_N) \in (\N\cup\{0\})^N$ for some $N \in \N \setminus \{0,1\}$, $K:=\sum\bfk$ and $B(\bfk):=K^K/\bfk^{\bfk}$. (It will help to keep Notation \ref{not:CombinationsOfVectors} in mind in the sequel.)

Consider the following Dirichlet series:
\begin{align*}
%\label{eqn:zetaBbfk}
\zeta_b(\bfk;s) &= \sum_{n=1}^\infty \binom{nK}{n\bfk} b^{-nKs}.
\end{align*} 
This form of Dirichlet series appears, indirectly, in Section \ref{sec:MultifractalAnalysisViaScalingRegularityAndScalingZetaFunctions}. Indeed, if the conditions of Theorem \ref{thm:FullFamilyOfScalingZetaFunctions} are satisfied and, additionally, if $N_0=N$, then we take $b^{-1} = \bfr^{\bfk/K} = \bft^{\bfm/M}$ with $K=M$ such that $\bfr = \bft$ is the vector of scaling ratios of the corresponding weighted self-similar system. We obtain
\begin{align*}
%\label{eqn:HypergeometricScalingZetaFunction}
\zeta_{\bfr,\bfp}(\beta(\bfm):s)	&= \sum_{n=1}^\infty\binom{nM}{n\bfm}\bft^{sn\bfm} = \sum_{n=1}^\infty \binom{nK}{n\bfk} b^{-nKs} = \zeta_b(\bfk;s).
\end{align*}

Now, via the formula (see, e.g., \cite{BatErd})
\begin{align*}
\binom{nK}{n\bfk} &=
\frac{\prod_{j=1}^K(1((j/K)+n)/1(j/K))}{\prod_{q=1}^N\left(\prod_{j=1}^{k_q}(1((j/k_q)+n)/1(j/k_q))\right)} \cdot \frac{B(\bfk)^n}{n!}
\end{align*}
we have
\begin{align}
\label{eqn:zetaBbfkHypergeometric1}
\zeta_b(\bfk;s) &= -1 + \sum_{n=0}^{\infty}\frac{\prod_{j=1}^K(1((j/K)+n)/1(j/K))}{\prod_{q=1}^N\left(\prod_{j=1}^{k_q}(1((j/k_q)+n)/1(j/k_q))\right)} \cdot \frac{(B(\bfk)b^{-Ks})^n}{n!}.
\end{align}
The series on the right-hand side of \eqref{eqn:zetaBbfkHypergeometric1} is a generalized hypergeometric series of the form $_{K}F_{K-1}$ (see \cite{BatErd} for the precise definition and notation). That is, we have
\begin{align*}
\zeta_b(\bfk;s) &= -1 + \sum_{n=0}^{\infty}\frac{\prod_{j=1}^K(1((j/K)+n)/1(j/K))}{\prod_{q=1}^N\left(\prod_{j=1}^{k_q}(1((j/k_q)+n)/1(j/k_q))\right)} \cdot \frac{(B(\bfk)b^{-s})^n}{n!} \label{eqn:zetaBbfkHypergeometric2} \\
&= -1 + _{K}F_{K-1}(1/K,\ldots,1; 1/k_1,\ldots,(k_n-1)/k_N; B(\bfk)b^{-Ks}). \notag
\end{align*}

Thus, even though the full families of scaling zeta functions of a self-similar measure in the general case have yet to be determined, the collection of scaling zeta functions which are given by a hypergeometric series is primed for further analysis. As seen in the next section, there are scaling zeta functions which are neither nearly the zeta function of a self-similar system nor a hypergeometric series, even in the case where the distinct scaling ratios are rationally independent.

\subsection{A self-similar system with nonlattice and lattice structure}
\label{sec:ASelfSimilarSystemWithNonlatticeAndLatticeStructure}

This last section investigates the structure of the scaling zeta functions in a special case where the distinct scaling regularity values are rationally independent but only two of the scaling zeta functions are known. 

Consider a weighted self-similar system $(\bfPhi,\bfp)$ such that $\bfPhi$ satisfies the open set condition and is nonlattice with scaling ratios $\bfr \in (0,1)^N$ for some $N \in \N$, $N \geq 3$, satisfying the following conditions: First, there exist $t,t_0 \in (0,1)$ and $\bft:=(t,t^2,t_0)$ such that for each $j=1,\ldots,N$, $r_j$ is equal to either $t,t^2$, or $t_0$. Second, there exists $\bfc=(c_1,c_2,c_0) \in \N^3$ such that $\sum\bfc=N$ and 
\begin{align*}
c_1	&=\#\{j:r_j=t, j \in \{1,\ldots,N\}\},\\
c_2	&=\#\{j:r_j=t^2, j \in \{1,\ldots,N\}\},
\end{align*}
and
\begin{align*}
c_0	&=\#\{j:r_j=t_0, j \in \{1,\ldots,N\}\}.
\end{align*}
Furthermore, suppose there are a probability vector $\bfp$ and rationally independent real numbers $\gamma$ and $\gamma_0$ such that for each $j \in \{1,\ldots,N\}$, $r_j=t$ implies $p_j=t^\gamma$; $r_j=t^2$ implies $p_j=t^{2\gamma}$; and $r_j=t_0$ implies $p_j=p_0:=t_0^{\gamma_0}$.

Due to Theorem \ref{thm:FullFamilyOfScalingZetaFunctions}, in order to determine the distinct scales $l_n(\alpha)$ and the corresponding multiplicities $m_n(\alpha)$ of the sequence of $\alpha$-scales $\calL_{\bfr,\bfp}(\alpha)$, it suffices to consider the scaling regularity values $\alpha(\bfk)$ where $\bfk \in (\N \cup \{0\})^N$, $\bfk \neq \mathbf{0}$, and $\gcd(\bfk)=1$. 

Fix $\bfk$ as above. Let
\begin{align*}
m	&:=\sum_{j:r_j=t}k_j+\sum_{j:r_j=t^2}k_j, \quad m_0:=\sum_{j:r_j=t_0}k_j,
\end{align*}
and $\bfm:=(m,m_0)$. We have 
\begin{align*}
\alpha(\bfk)	&:=\log_{\bfr^\bfk}\bfp^\bfk = \frac{\log t^{m\gamma}t_0^{m_0\gamma_0}}{\log t^{m}t_0^{m_0}}.
\end{align*}
Define $\beta(\bfm):=\log t^{m\gamma}t_0^{m_0\gamma_0}/\log t^{m}t_0^{m_0}$.

Now, for a fixed $\bfm$ where $\gcd(\bfm)=1$ and each $n \in \N$ we have
\begin{align*}
l_n(\beta(\bfm))	&=(t^m t_0^{m_0})^n.
\end{align*}

The key difficulty here lies in determining the multiplicities $m_n(\beta(\bfm))$ of the scales $l_n(\beta(\bfm))$, specifically, due to the fact that some initial scaling ratios are equal to $t$ and others are equal to $t^2$. Thus, for instance, the $\gamma$-scale $t^4$ is attained with respect to some vectors $\bfk_1, \bfk_2,$ and $\bfk_3$ where $|\bfk_1|=2$, $|\bfk_2|=3$, and $|\bfk_3|=4$. More specifically, if both nonzero components of $\bfk_1$ correspond to the scale $t^2$, two of the nonzero components of $\bfk_2$ correspond to $t$ and the other corresponds to $t^2$, and all four nonzero components of $\bfk$ correspond to $t$, then 
\[
\alpha(\bfk_1)=\alpha(\bfk_2)=\alpha(\bfk_3)=\beta(\bfm)=\gamma
\] 
and 
\[
\bfr^{\bfk_1}=\bfr^{\bfk_2}=\bfr^{\bfk_3}=t^4.
\] 

That is, a given scale $l_n(\beta(\bfm))$ can arise in various stages with respect to $(\bfPhi,\bfp)$, making the determination of the precise form of the multiplicity $m_n(\beta(\bfm))$ difficult in general. As a result, and although the full family of scaling zeta functions in the setting of this section have been determined, the abscissae of convergence and the $\alpha$-scaling complex dimensions are known for just two nontrivial scaling regularity values.

The following notation is used repeatedly below. Let $\left\lfloor \cdot \right\rfloor$ denote the floor function. That is, for $x \in \R$, $\left\lfloor x \right\rfloor$ is the integer part of $x$ given by the greatest integer such that $\left\lfloor x \right\rfloor \leq x$.

Suppose $\gcd(\bfm)=1$ and set $M=\sum\bfm$. Then, for each $n \in \N$,
\begin{align}
m_n(\beta(\bfm))	&=\#\{J \in \scrJ : r_J = l_n(\beta(\bfm))\} \notag \\
	&=\sum_{j=0}^{\left\lfloor nM/2 \right\rfloor}\binom{\sum\bfv(j)}{\bfv(j)}\bfc^{\bfv(j)}, \label{eqn:BetamMultiplicity}
\end{align}
with $\bfv(j)=(v_1(j),v_2(j),v_0)$ where, for $j=0,\ldots,\left\lfloor nM/2 \right\rfloor$,
\begin{align*}
v_1(j)	&:=nm-2\left\lfloor nM/2 \right\rfloor+2j, \quad v_2(j):=\left\lfloor nM/2 \right\rfloor-j, \quad \textnormal{and } v_0:=nm_0.
\end{align*} 

The difficulty in determining the form of the generalization of $m_n(\beta(\bfm))$ lies in determining the relationship between the components of the generalization of the vectors $\bfv(j)$. The $\beta(\bfm)$-scaling zeta function of $(\bfPhi,\bfp)$ for some $\bfm$ with $\gcd(\bfm)=1$ is given by 
\begin{align*}
\zeta_{\bfr,\bfp}(\beta(\bfm);s)	&=\sum_{n=1}^\infty m_n(\beta(\bfm))(l_n(\beta(\bfm)))^s,
\end{align*}
where $\textnormal{Re}(s)$ is large enough.

\begin{example}
\label{eg:SpecialCasesOfGeneralizedScottStuff}
In this example, the closed forms of the scaling zeta functions for just two nontrivial scaling regularity values, namely $\gamma_0$ and $\gamma$, are known. Both cases are the result of Corollary \ref{cor:NearlyRecoverGeneralizedLatticeString}, thus the abscissae of convergence and the $\alpha$-scaling complex dimensions associated with just two nontrivial regularity values can be obtained.

We have $\gcd(\bfm)=1$ and $\beta(\bfm)=\gamma_0$ if and only if $m=0$ and $m_0=1$. In this case, $l_n(\gamma_0)=t_0^n$, $m_n(\gamma_0)=c_0^n$, and by Corollary \ref{cor:NearlyRecoverGeneralizedLatticeString} we have, for $s \in \C,$
\begin{align*}
\zeta_{\bfr,\bfp}(\gamma_0;s)	&= \frac{c_0t_0^s}{1-c_0t_0^s}.
\end{align*}
Moreover,
\begin{align*}
%\label{eqn:Gamma0ComplexDimensions}
\calD_{\bfr,\bfp}(\gamma_0)	&=\left\{-\log_{t_0}c_0 + i\frac{2\pi}{\log{t_0}}z : z \in \Z\right\}.
\end{align*}

As for scaling regularity $\gamma$, we have $\gcd(\bfm)=1$ and $\beta(\bfm)=\gamma$ if and only if $m=1$ and $m_0=0$. In this case, $l_n(\gamma)=t^n$ and 
\begin{align*}
m_n(\gamma)	&=\sum_{j=0}^{\left\lfloor n/2 \right\rfloor}\binom{\left\lfloor \frac{n+1}{2}+j \right\rfloor}{n-2\left\lfloor \frac{n}{2} \right\rfloor+2j}\bfc^{\bfv(j)}.
\end{align*}
However, a closed form of $\zeta_{\bfr,\bfp}(\gamma;s)$ is obtained in Corollary \ref{cor:NearlyRecoverGeneralizedLatticeString}. That is,
\begin{align*}
\zeta_{\bfr,\bfp}(\gamma;s)	&= \frac{c_1t_1^s+c_2t_2^{2s}}{1-c_1t_1^s-c_2t_2^{2s}}, \quad s \in \C.
\end{align*}
Moreover, $\zeta_{\bfr,\bfp}(\gamma;s) = \zeta_\eta(s)$, where $\eta$ is the generalized lattice string with scaling ratios $r_1=t$ and $r_2=t^2$, weights $m_1=c_1$ and $m_2=c_2$, and $c_1$ gaps $g_1=t$ and $c_2$ gaps $g_2=t^2$. (See Definition \ref{def:GeneralizedLatticeString}.) Hence, $\calD_{\bfr,\bfp}(\gamma)=\calD_\eta$.
\end{example}

\begin{example}
\label{eg:RecoverScott}
This example develops the full recovery of the Fibonacci string as in Example \ref{eg:TheFibonacciString} as well as the results of Example 5.16 from \cite{ELMR}. 

Consider a weighted self-similar system $(\bfPhi,\bfp)$ such that $\bfPhi$ satisfies the open set condition and has scaling ratios $\bfr=(1/2,1/4,1/10)$. Then $\bfPhi$ is nonlattice. Furthermore, suppose the probability vector $\bfp$ is given by $\bfp=(1/2,1/4,1/4)$ such that the distinct initial scaling regularity values $\gamma=1$ and $\gamma_0=\log_{10}{4}$ are rationally independent. Moreover, in the notation given by the more general setting above, we have $t=1/2, t_0=1/10$ and $c_1=c_2=c_0=1$. As in the previous example, the only scaling regularity values for which the scaling zeta functions are known are $\gamma_0=\log_{10}4$ and $\gamma=1$, so only these two cases are discussed in this example.

First consider the simplest case of scaling regularity $\gamma_0=\log_{10}4$. The $\gamma_0$-scaling zeta function is given by
\begin{align*}
%\label{eqn:ScottEasyScalingZetaFunction}
\zeta_{\bfr,\bfp}(\gamma_0;s)	&= \sum_{n=1}^{\infty}10^{-ns} = \frac{10^{-s}}{1-10^{-s}}.
\end{align*}
Thus, the set of $\gamma_0$-scaling complex dimensions is given by
\[
\calD_{\bfr,\bfp}(\gamma_0)=\left\{izp : z \in \Z \right\},
\]
where $p=2\pi/\log{10}$.

In this case of scaling regularity $\gamma=1$, we nearly recover the geometric zeta function of the Fibonacci string $\zeta_{\textnormal{Fib}}$. The $1$-scales are given by $l_n(1)=2^{-n}$ for $n \in \N$ and the multiplicities $m_n(1)$ are given by the Fibonacci numbers. In fact, in light of the formula for the more general multiplicities $m_n(\beta(\bfm))$ in \eqref{eqn:BetamMultiplicity} and the formulas from Example \ref{eg:TheFibonacciString}, we recover a classic formula which yields the Fibonacci numbers as sums of particular binomial coefficients. Specifically, 
\begin{align*}
%\label{eqn:RecoverFibonacciNumbers}
m_n(1)	&= \sum_{j=0}^{\left\lfloor\frac{n}{2}\right\rfloor} \binom{\left\lfloor\frac{n+1}{2}+j\right\rfloor}{n-2\left\lfloor\frac{n}{2}\right\rfloor+2j}=F_{n+1},
\end{align*}
where $F_{n+1}$ is the $(n+1)$th Fibonacci number. See Example \ref{eg:TheFibonacciString} above, Example 5.16 of \cite{ELMR}, and \cite[\S2.3.2]{LapvF6} for a discussion of the Fibonacci string and its geometric zeta function $\zeta_{\textnormal{Fib}}$. 

The scaling zeta function $\zeta_{\bfr,\bfp}(1;s)$ is then given by
\begin{align*}
%\label{eqn:RecoverFibonacciGZF}
\zeta_{\bfr,\bfp}(1;s)	&=\sum_{n=1}^{\infty}F_{n+1}2^{-ns} 		=\zeta_{\textnormal{Fib}}(s)-1=\frac{2^{-s}+4^{-s}}{1-2^{-s}-4^{-s}},
\end{align*}
where $\zeta_{\textnormal{Fib}}$ is given by \eqref{eqn:GeometricZetaFunctionFibonacci}. Thus, the corresponding complex dimensions, in both the classic sense and with respect to the scaling regularity $1$, are given by \eqref{eqn:FibonacciComplexDimensions}. That is,
\begin{align*}
%\label{eqn:FibonacciComplexDimensionsAgain}
\calD_{\textnormal{Fib}}	&=\calD_{\bfr,\bfp}(1)
=\left\{D_{\textnormal{Fib}}+izp : z \in \Z \right\} \cup \left\{-D_{\textnormal{Fib}}+i(z+1/2)p : z \in \Z \right\},
\end{align*}
where $\varphi=(1+\sqrt{5})/2$ is the Golden Ratio, $D_{\textnormal{Fib}}=\log_2{\varphi}$, and the oscillatory period is $p = 2\pi/\log{2}$. 
\end{example}

In closing, we mention that some next steps include the determination of the $\alpha$-scaling complex dimensions and the full tapestry of complex dimensions $\calT_{\bfr,\bfp}$ associated with such self-similar measures (cf.~\cite[\S 6]{ELMR}). Thereafter, one may study the implications (for the oscillatory behavior of self-similar systems and measures) of counting functions and volume formulas associated with the nontrivial scaling regularity values, in the spirit of similar notions stemming from the theory of complex dimensions of fractal strings developed throughout \cite{LapvF6} (see also \cite{LapPeWin} and the relevant references therein, for the higher-dimensional case).

\bibliographystyle{amsalpha}

%\addcontentsline{toc}{section}{References}
%\section{References}

%\include{MASZFRSLS_bib}

\end{document}